\def\COLTmode{0}

\ifnum\COLTmode=1
\documentclass{colt2016}
\else
\documentclass[11pt,letterpaper]{article}
\fi
\usepackage{times}

\def\showauthornotes{0}

\usepackage
{
   amssymb, amsmath
}

\ifnum\COLTmode=0
\newcommand{\COLT}[1]{}
\newcommand{\notCOLT}[1]{#1}
\else
\newcommand{\COLT}[1]{#1}
\newcommand{\notCOLT}[1]{}
\fi

\ifnum\COLTmode=0
\usepackage{amsthm}
\usepackage[round]{natbib}
\usepackage{color,xspace}
\usepackage{fullpage}
\usepackage[unicode,colorlinks=true,citecolor=blue,filecolor=black,linkcolor=blue,urlcolor=black, pdfpagelabels, plainpages=false]{hyperref}

\newcommand {\acks}[1]{\section*{Acknowledgement}#1}

\newtheorem{theorem}{Theorem}

\newtheorem*{theorem*}{Theorem}

\newtheorem*{proposition*}{Proposition}
\newtheorem{lemma}[theorem]{Lemma}
\newtheorem*{lemma*}{Lemma}
\newtheorem{corollary}[theorem]{Corollary}
\theoremstyle{definition}
\newtheorem{definition}[theorem]{Definition}

\newcommand*{\doi}[1]{}
\fi

\newcommand{\OpenFrame}{\notCOLT{\rule{0pt}{12pt} \hrule height 0.8pt \rule{0pt}{1pt} \hrule height 0.4pt}\vspace{-4pt}\COLT{\medskip}}
\newcommand{\CloseFrame}{\notCOLT{\vspace{-6pt}\hrule height 0.4pt \rule{0pt}{1pt} \hrule height 0.8pt \rule{0pt}{12pt}}}

\newtheorem{fact}[theorem]{Fact}
\newtheorem*{fact*}{Fact}

\newtheorem*{hypothesis*}{Hypothesis}
\newtheorem{claim}[theorem]{Claim}
\newtheorem*{claim*}{Claim}
\newtheorem*{remark*}{Remark}

\newtheorem*{observation*}{Observation}

\usepackage{nicefrac}

\newcommand{\eps}{\epsilon}

\ifnum\showauthornotes=1
\newcommand{\Authornote}[2]{{\sffamily\small\color{red}{[#1: #2]}}}
\newcommand{\Authornotecolored}[3]{{\sffamily\small\color{#1}{[#2: #3]}}}
\newcommand{\Authorcomment}[2]{{\sffamily\small\color{gray}{[#1: #2]}}}
\newcommand{\Authorstartcomment}[1]{\sffamily\small\color{gray}[#1: }

\newcommand{\Authorfnote}[2]{\footnote{\color{red}{#1: #2}}}
\newcommand{\Authorfixme}[1]{\Authornote{#1}{\textbf{??}}}
\newcommand{\Authormarginmark}[1]{\marginpar{\textcolor{red}{\fbox{\Large #1:!}}}}
\else
\newcommand{\Authornote}[2]{}
\newcommand{\Authornotecolored}[3]{}
\newcommand{\Authorcomment}[2]{}
\newcommand{\Authorstartcomment}[1]{}

\newcommand{\Authorfnote}[2]{}
\newcommand{\Authorfixme}[1]{}
\newcommand{\Authormarginmark}[1]{}
\fi

\newcommand {\roundup}   [1] {{\lceil {#1} \rceil}}
\newcommand {\rounddown} [1] {{\lfloor {#1} \rfloor}}

\newcommand{\Set}[1]{\left\{#1\right\}}

\newcommand{\norm}[1]{\lVert#1\rVert}

\newcommand{\Iprod}[1]{\left\langle#1\right\rangle}

\newcommand{\Esymb}{\mathbb{E}}
\newcommand{\Psymb}{\mathbb{P}}
\newcommand{\Vsymb}{\mathbb{V}}

\DeclareMathOperator*{\E}{\Esymb}
\DeclareMathOperator*{\Var}{\Vsymb{}ar}
\DeclareMathOperator*{\ProbOp}{\Psymb}
\DeclareMathOperator{\core}{core}

\renewcommand{\Pr}{\ProbOp}

\newcommand{\given}{\mid}

\newcommand{\textparen}[1]{\text{(#1)}}

\ifx\because\undefined
\newcommand{\because}[1]{\textparen{because #1}}
\else
\renewcommand{\because}[1]{\textparen{because #1}}
\fi

\newcommand{\symdiff}{\Delta}

\DeclareMathOperator{\sdp}{sdp}
\DeclareMathOperator{\plant}{planted}

\DeclareMathOperator{\avg}{Avg}
\DeclareMathOperator*{\aavg}{Avg}

\newcommand{\Erdos}{Erd\H{o}s\xspace}
\newcommand{\Renyi}{R\'enyi\xspace}

\newcommand{\cG}{\mathcal G}

\newcommand{\calE}{{\cal E}}

\newcommand{\bbZ}{\mathbb Z}

\let\epsilon=\varepsilon

\numberwithin{equation}{section}

\newcommand{\dis}[2]{\tfrac{1}{2}\norm{#1 - #2}^2}

\newcommand{\Kg}{K_G}
\newcommand{\ubar}{\bar u}
\newcommand{\vbar}{\bar v}
\newcommand{\SBM}{\text{SBM}}
\newcommand{\dkl}[1]{d_{\text{KL}} \left( #1 \right)}

\newcommand{\constsparse}{c_{\ref{thm:distances}}}
\newcommand{\constsparseB}{c_{\ref{cor:distances}}}

\date{}

\ifnum\COLTmode=1
\title[Learning Communities in the Presence of Errors]{Learning Communities in the Presence of Errors}
 \coltauthor{\Name{Konstantin Makarychev }\addr Microsoft Research \AND \Name{Yury Makarychev }\addr TTIC \AND \Name{Aravindan Vijayaraghavan }\addr Northwestern University}
\else
\title{Learning Communities in the Presence of Errors}
\author{Konstantin Makarychev\\Microsoft Research \and Yury Makarychev\\TTIC \and Aravindan Vijayaraghavan\\Northwestern University}
\fi
\begin{document}

\maketitle

\begin{abstract}

We study the problem of learning communities in the presence of modeling errors and give  robust recovery algorithms for the Stochastic Block Model (SBM). This model, which is also known as the Planted Partition Model,  is widely used for community detection and graph partitioning in various fields, including machine learning,  statistics, and social sciences.
Many algorithms exist for learning communities in the Stochastic Block Model, but they do not work well in the presence of errors.

In this paper, we initiate the study of robust algorithms for partial recovery in SBM with modeling errors or noise.
We consider graphs generated according to the Stochastic Block Model and then modified by an adversary. We allow two types of adversarial errors,
Feige--Kilian or monotone errors, and edge outlier errors. Mossel, Neeman and Sly (STOC 2015) posed an open question about whether an almost exact recovery is possible when the adversary is allowed to add $o(n)$ edges. Our work answers this question affirmatively even in the case of $k>2$ communities.

We then show that our algorithms work not only when the instances come from SBM, but also work when the instances come from
any distribution of graphs that is $\varepsilon m$ close to SBM in the Kullback--Leibler divergence.
This result also works in the presence of adversarial errors. Finally, we present almost tight lower bounds for two communities.
 \end{abstract}


\section{Introduction}

Probabilistic models are ubiquitous in machine learning and widely used to find hidden structure in unlabeled data.
The Stochastic Block Model (SBM), which is also known as Planted Partition Model, is the most studied probabilistic model for community detection
and graph partitioning.
There has been extensive research on the model in various fields, including machine learning,  statistics, computer science, and social sciences over the last three decades (this research is summarized in Section~\ref{sec:prior}). Until recently, research on SBM was focused on graphs with  a
poly-logarithmic, in the number of vertices, average degree. In the past few years, however, most of the research has shifted toward graphs with a constant average degree, and there has been
significant progress in the understanding of the conditions under which a partial recovery is possible for such graphs in SBM.
In particular, \cite{M14} and \cite{MNS12,MNS13} have derived sharp conditions under which
a partial recovery is possible for the case of two communities (clusters).

Yet most existing algorithms are not robust --- they rely on the instance being drawn exactly from the
given probabilistic model, and thus may fail in the presence of noise.
For instance, while spectral algorithms have good provable guarantees for learning communities in SBM,
they crucially rely on strong spectral properties of random graphs, which are brittle to a small amount of noise.

Algorithms most commonly employed in practice, for learning various probabilistic models are based on maximum likelihood estimation.
They have many desirable properties from a statistical standpoint, since maximum likelihood estimation is robust to many modeling errors.
 However, they do not typically have polynomial running time guarantees.
This leads to a natural question: \textit{Can we design algorithms for learning communities in SBM, which are both efficient (polynomial time) and tolerant to adversarial modeling errors?}

In this paper, we present polynomial-time algorithms that perform robust recovery for the Stochastic Block Model (SBM).
Our algorithms work in the presence of different types of adversarial noise: edge outlier errors, monotone errors, and a modeling error measured in the Kullback--Liebler divergence. Our results
give an affirmative answer to the question posed by \cite{MNS15},
 whether an almost exact recovery is possible when the adversary is allowed to add $o(n)$ edges.

Let us now recall the definition of the Stochastic Block Model\footnote{We note that some papers denote by $n$ not the number of vertices in each cluster but the total number of vertices. Our $\SBM(n, k, a, b)$ model is
the same as their $\SBM'(kn, k, ka, kb)$ model.}.
\begin{definition}[Stochastic Block Model]\label{def:sbm}
A graph $G_{sb}(V,E_{sb})$ with $N=nk$ vertices is generated according to the Stochastic Block Model  $\SBM(n,k,a,b)$ (where $a \ge b$) as follows:
\begin{enumerate}
\item There is a equipartition $P^*=\left(V^*_1, V^*_2, \dots, V^*_k\right)$ of vertices $V$ with $|V^*_i|=n$ for each $i \in [k]$.
\item For each $i \in [k]$, and for any two vertices $u, v \in V^*_i$, there is an edge $(u,v) \in E_{sb}$ with probability $a/n$.
\item For each $i, j \in [k]$ with $i \ne j$, and for any two vertices $u \in V^*_i, v \in V^*_j$, there is an edge $(u,v)\in E_{sb}$ with probability $b/n$.
\end{enumerate}
We denote the expected number of edges in $G$ by $m$: $m= \tfrac{1}{2}(nka+n k(k-1)b)$.
\end{definition}
We consider the Stochastic Block model with two types of modeling errors (adversarial noise): the outlier errors and
Feige--Kilian~(\citeyear{FKil}) (monotone) errors.
\begin{definition}[Stochastic Block Model with modeling errors]\label{def:outliers}
In the Stochastic Block Model $\SBM(n,k,a,b)$ with modeling errors, the graph $G(V,E)$  is generated as follows.
First, a random graph $G_{sb}=(V, E_{sb})$ is sampled from the Stochastic Block Model $\SBM(n,k,a,b)$.
Then the adversary adds some new edges to $E'$ and removes some existing edges from $E'$. Specifically,
the adversary may do the following:
\begin{enumerate}
\item In \textit{the Feige--Kilian or monotone error model}, the adversary may add any edges within the clusters and remove any clusters between the clusters.
\item In \textit{the model with $\varepsilon m$ outliers}, the adversary may choose $\varepsilon_1\geq 0$ and $\varepsilon_2 \geq 0 $ with $\varepsilon_1 + \varepsilon_2 \leq
 \varepsilon$, then add at most $\varepsilon_1 m$ edges between the clusters and remove at most $\varepsilon_2 m$ edges within the clusters.
 \item In the model with two types of errors, the adversary may introduce both types of errors.
 \end{enumerate}
\end{definition}
Our goal is to find the unknown planted partition $(V_1^*, \dots, V_k^*)$ given the graph $G = (V,E)$ from the Stochastic Block Model with modelling errors. However, in this paper,
we focus on the regime where the exact recovery is impossible even information--theoretically. So we are interested in designing polynomial--time algorithms that partially recover the planted partition.
\begin{definition}\label{def:close-and-recovery}
We say that a partition $V_1, \dots, V_k$ is $\delta$-close to the planted partition $V_1^*,\dots, V_k^*$, if each cluster $V_i$ has size exactly $n$ and there is a permutation $\sigma$ of indices such that $$\Bigl|\bigcup_{j = \sigma(i)}   V_i^* \cap V_j\Bigr| \geq (1 - \delta) kn.$$
An algorithm $(1-\delta)$-partially recovers the planted partition if it finds a partition that is $\delta$--close to the planted partition.
\end{definition}

We present two algorithms for partial recovery. The first algorithm can handle instances with both monotone and outlier errors, while the second algorithm
handles only instances with outlier errors. The second algorithm also has stronger requirements on $a$ and $b$. However, it has a much better recovery guarantee.
\begin{theorem}[First Algorithm]\label{thm:main}
Consider the stochastic block model $\SBM(n,k,a,b)$ with $\varepsilon m$ outliers and monotone errors, and suppose $a+b(k-1) \ge C_0$ for some universal constant $C_0> 1$.
There is a polynomial-time algorithm that $(1-\delta)$-partially recovers the planted partition given an instance of the model, where
$$\delta = O\Bigl(\frac{ \sqrt{a+b(k-1)} }{a-b}+\frac{\epsilon\left(a+b(k-1)\right)}{a-b}\Bigr).$$
The algorithm succeeds with probability at least $1-2\exp\left(-2N\right)$ over the randomness of the instance.

 Furthermore, for any $\eta \in (1/(a+b(k-1)),\tfrac{1}{2})$, with probability at least
 $1-2\exp(-\eta m)$, the algorithm $(1-\delta')$-partially recovers the planted partition
 with
 $$\delta' = O\Bigl(\frac{(\epsilon+\sqrt{\eta})\left(a+b(k-1)\right)}{a-b}\Bigr).$$
\end{theorem}
For the case of two communities ($k=2$), we prove that the result of Theorem~\ref{thm:main-amplified} is asymptotically
optimal (see Theorem~\ref{thm:lb}). We also note that in the special case of $k=2$ communities, the analysis
of the algorithm due to \cite{GV14} can be adapted to obtain similar results (up to constants). However, their approach
breaks down for $k\ge 3$ communities (see Section~\ref{sec:techniques} for details).

\begin{theorem}[Second Algorithm]\label{thm:main-amplified}
Consider the stochastic block model $\SBM(n,k,a,b)$ with $\varepsilon m$ outliers (without any monotone modelling errors), and suppose $a+b(k-1) \ge 2C_0$ for some universal constant $C_0 >1$. Assume that
$$\frac{ \sqrt{a+b(k-1)} }{a-b}+\frac{\epsilon\left(a+b(k-1)\right)}{a-b} \leq c/k,$$
where $c>0$ is some absolute constant.
There is  a randomized polynomial-time algorithm that does the following.
Let $\delta_0 \geq ke^{-\frac{(a-b)^2}{100 a}}$ and $\delta = O(\delta_0 + \frac{\varepsilon m}{(a-b) kn})$.
The algorithm $(1-\delta)$-partially recovers the planted partition
with probability at least $1-3\exp(- \delta_0 k n/6)$ over the randomness of the instance and random bits used by the algorithm.
\end{theorem}
In the above theorems $C_0$ is some universal constant that lower bounds the average degree of a vertex ($C_0=11$ suffices). We do not make any efforts to optimize the constant $C_0$, for ease of exposition.
Let us compare the performance of our algorithms to the performance of the state of the art algorithms for the Stochastic Block Model.
\begin{itemize}
\item If no adversarial noise is present, our first algorithm works under the same condition on parameters $a$, $b$ and $k$:
$$\frac{(a-b)^2}{a+b(k-1)} > C \quad \text{for some absolute constant } C$$
as the algorithm by~\cite{AS15} for SBM (the absolute constant $C$ in our condition is different from that by~\cite{AS15}).
\item Our second algorithm achieves the same recovery rate as the algorithm of \cite{CRV15} for SBM
(that, however, is not surprising, since our second algorithm uses the ``boosting'' technique developed by~\cite{CRV15}) .
\end{itemize}

We note that, unlike many previously known algorithms for the Stochastic Block Model, our recovery algorithms fail with probability
that is exponentially small in $\eta m$.
In particular, this implies that the algorithm from Theorem~\ref{thm:main} works even if we sample the initial graph $G_{sb}$ not from $\SBM(n, k, a, b)$
but from a distribution that is $(\lambda m)$-close to $\SBM(n, k, a, b)$ in the KL-divergence distance (see Section~\ref{sec:KL}).

\begin{theorem}\label{thm:KL}
Let $\cal G$ be a distribution that is $\lambda m$ close to $\SBM(n, k, a, b)$ in the KL divergence:
\COLT{$D_{KL}({\cal G}, \SBM(n, k, a, b))\leq \lambda m$.}\notCOLT{$$D_{KL}({\cal G}, \SBM(n, k, a, b))\leq \lambda m.$$}
Suppose that $a+b(k-1)\ge C_0$ for some universal constant $C_0>1$. Consider a model where the graph is sampled from
the distribution $\cal G$ and then the adversary introduces monotone and outlier modeling errors (with parameter $\varepsilon m$). For any $\eta>0$, the algorithm from Theorem~\ref{thm:main} works in this model with the same recovery guarantee:
$$\delta = O\Bigl(\frac{(\epsilon+\sqrt{\eta})\left(a+b(k-1)\right)}{a-b} +\frac{\sqrt{a+b(k-1)}}{(a-b)}\Bigr).$$
It may fail with probability at most $2\lambda/\eta$.
\end{theorem}


\paragraph{Related Work}
\cite{CL15} proposed a stochastic block model with \emph{outlier vertices}. In their model, the graph is
generated as follows: first a graph is drawn according to $\SBM(n,k,a,b)$, then
the adversary adds to the graph $t$ extra vertices and an arbitrary set of edges incedent on these $t$ vertices.
\cite{CL15} give an SDP algorithm for partially recovering the communities. Their algorithm works for
$a \ge C \log n$. For $a, b= O(\log n)$, it can tolerate up to $O(\log n)$ vertex outliers.
Note that if $a+b(k-1) \geq C\log n$ (as in their result), robustness to edge errors is more general than
robustness to vertex outliers.

In a concurrent and independent work, \cite{MPW15} study the problem of weak recovery in a SBM with $k=2$ communities
in the presence of monotone errors as in~\cite{FKil}. They do not consider the case of $k > 2$ communities; they also do not
consider adversarial errors and modeling errors in the KL divergence.
\COLT{Due to the space limit, we}
\notCOLT{We}
give a  detailed overview
of related work (including \citep{CL15} and \citep{MPW15}) in Section~\ref{sec:prior}.
\enlargethispage*{0.1cm}

\subsection{Techniques} \label{sec:techniques}

Let us briefly describe our first algorithm.
The algorithms is based on semidefinite programming (SDP).
 We use a variant of the standard SDP relaxation for the $k$-Partitioning Problem (see e.g.~\cite{KNS}).
The SDP solution assigns a unit vector $\bar u$ to each vertex $u$ of the graph (see Section~\ref{sec:prelim} for details).
We prove that vectors $\{\bar u\}$ are clustered consistently with the community memberships: the vectors assigned to vertices
in the same cluster are close to each other (on average), while the vectors assigned to vertices in different clusters are far from each other (on average).
It follows that  each cluster $V_i^*$ has a core $\core(i)$ such that all vertices in the core lie close to each other, vertices in different cores are
far apart, and $\cup_i \core(i)$ contains all but a small fraction of the vertices (see Section~\ref{sec:clustering}).

We give a simple greedy algorithm that, given the SDP solution, finds a partition $V_1^*,\dots, V_k^*$ of $V$ close to the planted
partition.
 The algorithm considers balls of some fixed small radius around vectors $\{\bar u\}$ and chooses the ``heaviest'' among them,
 the one that contains most vectors $\{\bar u\}$. It creates a cluster consisting of the vertices, whose vectors lie in the ball,
 and removes them and the corresponding vectors from the consideration. Then it iteratively processes the remaining vertices.
The clustering algorithm is similar to the algorithm
recently developed for a different clustering problem called Correlation Clustering \citep{MMVCC}. Unlike the algorithm in~\cite{MMVCC}, however, the algorithm in this work is robust to adversarial errors and modeling errors, and works in the sparse regime.
Importantly, our geometric structural property holds even in the presence of adversarial noise,
and the probability that a random graph from SBM does not satisfy it is exponentially small.
As a result of this, the algorithm works even in the presence of outlier, monotone, and modeling errors.

 To prove that our geometric structural property holds, we use, in particular, some techniques developed by
 \cite{GV14}.
 However, we cannot merely rely on their result: Gu\'edon and Vershynin prove that the best rank-$(2k-3)$ approximation $\hat P$ to the SDP solution matrix $\hat Z$ (the Gram matrix of the SDP vectors $\{\bar u\}$) is
 close to a particular rank-$(k-1)$ matrix, the matrix that encodes the planted partition.
 This property suffices when $k=2$ --- then the rank-1 matrix $\hat P$ defines a one dimensional solution
 $\{x_u: u\in V\}$, and the planted partition can be approximately recovered by thresholding numbers $\{x_u\}$.
 Moreover, it can be shown that this algorithm
for $k=2$ communities is robust to modeling errors.
However, this approach works only when $k=2$ and does not seem to extend to the case of $k>2$
(in particular, Gu\'edon and Vershynin only describe an algorithm for the case of $k=2$).
Therefore, instead of directly using the result by \cite{GV14}, we use some of their ideas to prove that
the SDP solution satisfies the geometric structural property (described above), which is quite different from that in~\cite{GV14}.
This property enables us to easily recover the planted clustering.


 \paragraph{Organization}
 We start by
presenting our SDP relaxation for the partition recovery problem (see Section~\ref{sec:prelim}).
Then, in Section~\ref{sec:distances} we prove the geometric structural property.
 In Section~\ref{sec:clustering}, we present our first algorithm and prove Theorem~\ref{thm:main}.
  In Section~\ref{sec:boosting}, we show how to ``boost'' the performance of this algorithm by using the technique by
  \cite{CRV15}. This yields
Theorem~\ref{thm:main-amplified}.
We present Theorem~\ref{thm:KL} in Section~\ref{sec:KL} and describe our negative results in Appendix~\ref{sec:lb}.
We give a  detailed overview of prior work in Section~\ref{sec:prior}.

\section{Overview of Prior Work}\label{sec:prior}
We now review prior work on learning probabilistic models for graph partitioning while focusing on algorithms that give polynomial time guarantees. In what follows, $C$ denotes a constant that is chosen to be sufficiently large.

\paragraph{Stochastic Block Models}
The Stochastic Block Model is the most widely studied probabilistic model for community detection and graph partitioning in different fields like machine learning, computer science,  statistics and social sciences (see e.g. ~\cite{BCLS84, stoch1, stoch2, Fortunato}). This model is also sometimes called the Planted Partitioning model and was studied in a series of papers, which among others include
\cite{DF87}, \cite{B88}, \cite{JS}, \cite{DI98}, \cite{CK01}, \cite{McSherry} and \cite{coja}. The existing algorithmic guarantees for the Stochastic Block Model fall into three broad categories: \emph{ exact recovery, weak recovery, and  partial recovery}.

For {\em exactly recovering} the communities, provable guarantees are known for many different algorithms like spectral algorithms, convex relaxations and belief propagation.  These algorithms need sufficient difference between the average intra-cluster degree $a$ and inter-cluster degree $b$, and a lower bound on the average degree $a+b= \Omega( \log n)$. For $k=2$ clusters,
\cite{B88} used spectral techniques to give an algorithm that recovers the clusters when $a-b \ge C \cdot \sqrt{a \log n}$. Recently,
\cite{ABH14} and \cite{MNS15} determined sharp thresholds for exact recovery in the case of $k=2$ communities. The influential work of
\cite{McSherry} used spectral clustering to handle a more general class of stochastic block models with many clusters, and the guarantees have been subsequently improved in different parameter regimes of $a,b,k$ by various works using both spectral techniques and convex relaxations \citep{Yudong, Ames, Vu, WXH15, PerryW15}.

The goal in {\em weak recovery} is to output a partition of the nodes which is positively correlated with the true partition with high probability. This problem was introduced by \cite{Co10}.
\cite{DKMZ11} conjectured that there is a sharp phase transition in the case of $k=2$ clusters depending on whether value of $\frac{(a-b)^2}{(a+b)}> 1$ or not, and this was settled independently by
\cite{M14} and \cite{MNS12,MNS13}. It was also recently shown that semidefinite programs get close to this
threshold \citep{MS15}\footnote{The algorithm of \cite{MNS13} and \cite{M14} uses non-backtracking random walks.}. The problem is still open for  $k>2$ communities, and
the conjecture of \cite{DKMZ11} and \cite{MNS13} for larger $k$ is that the clustering problem can be solved in polynomial time when $\frac{(a-b)^2}{a+(k-1)b} > 1$.


In {\em partial recovery}, the goal is to recover the clusters in the planted partitioning up to $\eta N$ vertices, i.e. up
 to $\eta N$ vertices are allowed to be misclassified in total  (here $\eta$ can be thought of as $o(1)$).
\cite{Co10} and \cite{MNS14colt} studied this problem for the case of $k=2$ communities.
\cite{GV14} analysed the semidefinite programming relaxation using the Grothendieck inequality to partially recover the
communities (for $k=2$) when $(a-b)^2 > C (a+b)/{\eta^2}$. These results were
extended to the case of $k$-communities by \cite{CRV15} and \cite{AS15}. The algorithm by \cite{CRV15} recovers the communities
up to $\eta$ error when $\frac{(a-b)^2}{a} \ge C k^2  \log (1/\eta)$.\footnote{In fact \cite{CRV15} gives the stronger guarantee
of recovering each of the clusters up to $\eta n$ vertices.} These results were recently improved by \cite{AS15} who gave algorithms and
information-theoretic lower bounds for partial recovery in fairly general stochastic block models.

We note that the algorithm and analysis of~\cite{GV14} can be adapted to work in the presence of monotone
and adversarial errors for the case of $k=2$ communities (see Section~\ref{sec:techniques} for details).
In a concurrent and independent work,~\cite{MS15} (see revision 2 of their archive paper)
observed that their algorithm for \emph{testing} whether the input graph comes from the \Erdos--\Renyi distribution
or a Stochastic Block Model with $k=2$ communities also works in presence of $o(m)$ edge outlier errors. Their algorithm
does not recover the clusters.

\paragraph{Semirandom models}
Semi-random models provide robust alternatives to average-case models by allowing much more structure then completely random instances.
Research on semi-random models was initiated by 
\cite{BS}, who introduced and investigated semi-ran\-dom models for $k$-coloring.
\cite{FKil} studied a semi-random model for Minimum Bisection (two communities of size $n$ each) that introduced the notion of a \emph{ monotone adversary}. The graph is generated in two steps: first a graph is generated according to $\SBM(n,2,a,b)$ and then an adversary is allowed to either add edges inside the clusters or delete some of the edges present between the clusters.
They showed that semi-definite programs remain integral when $a-b \ge C \cdot \sqrt{a \log n}$. This was also extended to the case of $k$ clusters
 by \cite{Yudong} and \cite{ABKK15, PerryW15}.

In a concurrent and independent work, \cite{MPW15} consider the problem of weak recovery in a SBM with $k=2$ communities in the presence of
monotone errors as in~\cite{FKil}. Their main result is a statistical lower bound that indicates that the phase transition
for weak recovery in SBM with $k=2$ communities changes in the presence of monotone errors. They also present an algorithm that performs
weak recovery for two communities in the presence of monotone errors.They do not consider the case of multiple communities ($k > 2$). They also
do not consider adversarial errors and modeling errors in the KL divergence.

The results by \cite{MMV,MMV2, MMVCC} use semi-definite programming to give algorithmic guarantees for various average-case models for graph
partitioning and clustering problems. These works~\citep{MMV, MMV2} consider probabilistic models for Balanced Cut (where the two clusters have roughly
equal size) that are more general than stochastic block models, but they are incomparable to the models considered in this work. Besides, the
focus of \citep{MMV, MMV2} is to find a Balanced Cut of small cost (the partitioning returned by the algorithm need not necessarily be
close to the planted partitioning) and they make no structural assumptions on the graph inside the clusters. The algorithm in \citep{MMV} also
 returns a partitioning closed to the planted partitioning under some mild assumptions about the expansion inside the clusters. However, it
 requires that $a = \tilde \Omega(\sqrt{\log n})$,
while the focus of this work is the regime when $a$ and $b$ are constants.

\paragraph{Handling Modeling Errors}
The most related result in terms of modeling robustness is the recent work by  \cite{CL15}, who consider the stochastic block model in
the presence of some outlier vertices. The graph is generated as follows: first a graph is drawn according to a stochastic
block model $SBM(n,k,a,b)$\footnote{The authors also consider the case where communities can have different sizes as well.}. Then,
the adversary adds to the graph $t$ outlier vertices and a set of arbitrary edges incedent on them.
\cite{CL15} give an SDP-based algorithm for partially
recovering the communities. Their algorithm works for
$a \ge C \log n$ and $(a-b) > C \left( \sqrt{a \log n} + \sqrt{kb} + m \sqrt{k}\right)$ (see Condition 3.1 in Theorem 3.1).
For $a, b= O(\log n)$, it can tolerate up to $O(\log n)$ outliers. To handle up to $\eps n$ outliers, the algorithm
needs the graph to be very dense i.e. $a, b = \Omega(\eps n)$.

In the regime when $a+b(k-1) \ge C \log n$, robustness to edge outliers is more general than robustness to vertex outliers. (Because,
in this regime, the degree of each vertex is tightly concentrated around $a+(k-1)b$, hence one can remove all outlier vertices whose degree
is substantially larger than $a+(k-1)b$ in the given graph $G$. After that the number of error edges will be $O(t (a+b(k-1)))$.
Using the results in our work, we can handle the case when an $\eps$ fraction
of the vertices are corrupted since this corresponds to an $\eps$ fraction of the edges being corrupted in {\em our} outlier model.
Additionally, our algorithm also performs partial recovery in the sparse regime (when $a,b = O(1)$).

\cite{KK10} and \cite{AS12} presented a spectral algorithm for clustering data that performs partial recovery as long the data satisfies some deterministic conditions (involving the spectral radius of the adjacency matrix), that are satisfied by instances that are generated by many probabilistic models for clusters.  These deterministic conditions hold in graphs with degree $\Omega(\log n)$ and when the noise is more structured; in particular, they need the spectral norm of a matrix representing the errors to be small (this does not hold for adversarial modeling errors in general).

Finally, the work of \cite{Brubaker} gave new algorithms for clustering data arising from a mixture of Gaussians when an $\eps =O(1/(k \log^2 n))$ fraction of the data points are outliers. Surprisingly, Brubaker showed that this tolerance to
noise can be achieved when the separation between the means is only a logarithmic factor more than the separation needed for learning gaussian mixtures with no
noise \citep{KSV05, AM05}. While these results apply to very different problems in unsupervised learning, in the analogous regime, our algorithm works if up to an $\eps = O(1)$
fraction of the observations come from errors.  Finally, our results also handle large errors in the probabilistic model, when measured in the KL divergence (up to $\eps m$).


\section{Preliminaries}\label{sec:prelim}
\subsection{Notation}
Given an  equipartition $(V^*_1, \dots, V^*_k)$ of the vertices of $G(V,E)$, let $(V \times V)_{in}$ represent all the pairs of vertices inside the clusters,
and $(V \times V)_{out}$ represent the pairs that go between the clusters. Similarly, let $E_{in}$ be the edges inside the clusters,
and $E_{out}$ be the edges that go between the different clusters.

\subsection{SDP Relaxation}
Our partition recovery algorithms are based on semidefinite programming. In all our algorithms,
we use the following basic SDP relaxation for the partition recovery problem (the SDP is presented in the vector form).
For every vertex $u$ in the graph, we have a vector variable $\bar u$ in the SDP relaxation.

\begin{align}
\min~& \sum_{(u,v) \in E} \dis{\ubar}{\vbar} &\\
\text{s.t.}& \notag\\
& \norm{\ubar}^2 =1 &~ \forall u \in V  \label{eq:unitvectorconstraint}\\
& \sum_{u,v \in V} \frac{1}{2} \norm{\ubar -\vbar}^2 =n^2k(k-1)=N^2 \left(1-\frac{1}{k}\right) \label{eq:balanceconstraint}\\
&\Iprod{\ubar,\vbar} \ge 0&~ \forall u,v\in V
\end{align}
The summation in constraint \eqref{eq:balanceconstraint} is over all $N^2$ pairs of vertices.

Our SDP relaxation is standard. Note that we do not use $\ell_2^2$-triangle
inequalities which are often used in SDP relaxations for graph partitioning problems.
We also do not use strong spreading constaints (see e.g.~\cite{KNS, BFK14}) and instead use
a weaker constraint~\ref{eq:balanceconstraint}.

We denote the optimal value of this SDP relaxation by $\sdp$. Consider the following feasible SDP solution corresponding to the planted partition. Assign $\ubar = e_i$ for all
$u\in V_i^*$ and all $i$, where $e_1,\dots, e_k$ is an orthonormal basis. It is easy to see that this is a feasible SDP solution. Its value is equal to the number of edges going
between partitions. Since the value of the optimal SDP solution is at most the value of this solution,
\begin{equation}
\sdp \leq |\{(u,v)\in E: u\in V_i^*, v\in V_j^* \text{ for some } i\neq j\}|. \label{ineq:sdp-less-than-opt}
\end{equation}

\newcommand{\fs}{e^{-\frac{9 s^2\,N}{4 + 8s/\sqrt{a+b(k-1)}}}}

\section{Structure of the Optimal SDP Solution}
\label{sec:distances}

In this section, we analyze the geometric structure of the optimal SDP solution. We show that SDP vectors for vertices in the same cluster are close
to each other (on average); SDP vectors for vertices in different clusters are far away from each other (on average).

We denote the average distances assigned by the SDP to pairs of vertices inside clusters and between clusters
by $\alpha$ and $\beta$, respectively. Formally,
$$
\alpha= \underset{(u,v) \in (V \times V)_{in}}\avg \dis{\ubar}{\vbar} \label{eq:alpha} \;\;\;\text{and}\;\;\;
\beta=\underset{(u,v) \in (V \times V)_{out}} \avg \dis{\ubar}{\vbar}. \label{eq:beta}
$$
It follows from constraint~(\ref{eq:balanceconstraint})  that the values of $\alpha$ and $\beta$ satisfy:
\begin{equation}\alpha + (k-1)\beta = k-1 .\label{eq:alpha-beta}\end{equation}
In the following theorem, we prove that $\alpha$ is small and $\beta$ is close to $1$.

\begin{theorem}\label{thm:distances}
Let $G(V,E)$ be a graph generated according to the stochastic block model $\SBM(n,k,a,b)$ with $\eps m $ outliers and
 arbitrary monotone errors. Suppose that $( a+b(k-1) )>C$ for some absolute constant $C$. Then, for every $s\geq 1$,
the average intra-cluster distance $\alpha$ and inter-cluster distance $\beta$ satisfy the following bounds with probability at least $1-2\fs$:
\begin{equation}\label{eq:avgdist}
\alpha \le \frac{ \constsparse (\sqrt{a+b(k-1)}) s}{a-b}+\frac{\left(a+b(k-1)\right)\epsilon}{a-b},
\end{equation}
and $\beta \geq 1 - \alpha /(k-1)$,
where  $\constsparse \le 6\Kg+4$ is an absolute constant and $\Kg < 1.783$ is the Grothendieck constant.
\end{theorem}
\begin{proof}
Denote $f(s) = \fs$. For notational convenience, we assume that all vertices in the graph are ordered.
Let $G_{sb}$ be the graph generated in the stochastic block model $SBM(n,k,a,b)$ without the adversarial errors; and let
$G$ be the graph obtained from $G_{sb}$ by introducing arbitrarily many monotone errors and at most $\varepsilon m$ non-monotone errors (here $m= (a+b(k-1))n/2$ is the expected number of edges in graphs from $\SBM(n,k,a,b)$).
Denote by $\plant (G)$ and $\plant (G_{sb})$ the cost of the planted partition in graphs $G$ and $G_{sb}$, respectively.
Denote by $\sdp (G_{sb},\{\tilde u\})$ the cost of a feasible SDP solution $\{\tilde u\}$ in the graph $G_{sb}$.
Let $\sdp(G)$ be the cost of the optimal SDP solution in $G$. Our goal
is to estimate $\plant (G) - \sdp (G)$. Note that the value of the SDP relaxation
is at most the value of the planted partition (see inequality~(\ref{ineq:sdp-less-than-opt})),
thus $\plant (G) - \sdp (G)\geq 0$.
We prove that with probability at least $1-2f(s)$,
\begin{equation}
\plant (G) - \sdp (G)\leq \frac{N \big(-\alpha (a-b) + \constsparse\sqrt{(a + (k-1)b)}s +  2\varepsilon m \big) }{2}.\label{eq:planted-sdp}
\end{equation}
This bound immediately implies the statement of the theorem: since $\sdp (G) \leq \plant(G)$, we have
$\alpha (a-b) \leq \constsparse\sqrt{(a + (k-1)b)}s +  2\varepsilon m$.
We first bound the value of $\plant (G_{sb}) - \sdp (G_{sb},\{\bar u\})$ for the graph $G_{sb}$,
where $\{\bar u\}$ is the optimal SDP solution for the graph $G$.
\begin{lemma}
The following inequality holds with probability at least $1-2f(s)$:
\begin{equation}
\plant (G_{sb}) - \sdp (G_{sb},\{\bar u\})\leq \frac{N \big(-\alpha (a-b) + \constsparse\sqrt{(a + (k-1)b)}s \big) }{2}.\label{eq:G-sb-plant-sdp}
\end{equation}
\end{lemma}
\begin{proof}
We upper bound $\plant (G_{sb})$. The expected size of the planted cut
equals $\E[\plant (G_{sb})] = bN(k-1)/2$. Thus, by the Bernstein inequality,
\begin{equation}
\plant(G_{sb}) \leq \frac{bN(k-1)}{2} + 2\sqrt{a+b(k-1)}Ns\label{eq:planted-Bernst}
\end{equation}
with probability at least $1-f(s)$ (see Lemma~\ref{lem:bernst-plant-cut} in Appendix~\ref{sec:appendix} for details).

We now lower bound $\sdp (G_{sb},\{\bar u\})$. Let $A=(a_{uv})$ be the adjacency matrix of $G$, and let $\E[A]$ be the expectation of the adjacency matrix. Denote $\Delta a_{uv} = a_{uv} - \E[a_{uv}]$. We use the following theorem, which is very similar to Lemma 4.1 in~\cite{GV14}. For completeness, we prove
Theorem~\ref{thm:fromGV} in Appendix~\ref{app:thm:fromGV}.
\begin{theorem}\label{thm:fromGV}
Let $G_{sb}(V,E)$ be a graph generated according to the stochastic block model $SBM(n, k, a, b)$.
Suppose $a+(k-1)b \geq 11$.
Then, with probability $1-f(s)$ the following inequality holds for all feasible SDP solutions $\{\tilde u\}$:
\begin{equation}
\Big| \sum_{u<v}
\Delta a_{uv} \|\tilde u -  \tilde v\|^2\Big| \leq  6K_G\sqrt{a+b(k-1)}Ns. \label{eq:sdp-GV}
\end{equation}
\end{theorem}
For the rest of the proof we assume that inequalities~(\ref{eq:planted-Bernst}) and (\ref{eq:sdp-GV}) hold. This happens with probability at least $1-2f(s)$.
We apply inequality~(\ref{eq:sdp-GV}) to the optimal SDP solution $\{\bar u\}$ for the graph $G$. We have
$$\sdp (G_{sb},\{\bar u\}) = \frac{1}{2}\sum_{u<v} a_{uv}\|\bar u - \bar v\|^2 \geq \frac{1}{2}\sum_{u<v} \E[a_{uv}]\|\bar u - \bar v\|^2
-3K_G \sqrt{a+b(k-1)}Ns.$$
The set of edges $E_{sb}$ comes from the stochastic block model, hence $\E[a_{uv}]=a/n$, if $(u, v) \in (V \times V)_{in}$; and
$\E[a_{uv}]=b/n$, if $(u,v) \in (V \times V)_{out}$. Therefore,
$$\frac{1}{2}\sum_{u<v} \E[a_{uv}]\|\bar u - \bar v\|^2 = \frac{a}{n} \sum_{\substack{(u,v)\in (V\times V)_{in} \\u<v}} \frac{\|\bar u - \bar v\|^2}{2} +
\frac{b}{n} \sum_{\substack{(u,v)\in (V\times V)_{out} \\u<v}} \frac{\|\bar u - \bar v\|^2}{2}.$$
By the definition of $\alpha$ and $\beta$, the first term on the right hand side equals $(a/n) \cdot \alpha k n^2/2 = a\alpha N/2$;
the second term equals $b \beta N (k-1) /2$. Using that $(k-1)\beta = (k-1) - \alpha$, we get
\begin{eqnarray*}
\sdp (G_{sb},\{\bar u^*\})&\geq& \frac{a\alpha N + b\beta (k-1) N}{2} -3K_G \sqrt{a+b(k-1)}Ns \\
&=& \frac{(a-b)\alpha N + b(k-1)N}{2}-3K_G \sqrt{a+b(k-1)}Ns.
\end{eqnarray*}
Combining this inequality with (\ref{eq:planted-Bernst}), we get bound~(\ref{eq:G-sb-plant-sdp}).
\end{proof}
Consider a sequence of operations -- edge additions and edge removals -- that transform the graph $G_{sb}$ into the graph $G$. Let $G_0 =G_{sb},\dots, G_T = G$ be the sequence of graphs obtained after performing these operations. Observe that every time we make a monotone change the value of $\plant (G_t) - \sdp (G_t,\{\bar u\})$ does not
increase: When we remove an edge between two vertices $u$ and $v$ in distinct clusters, we decrease $\plant(G_t)$ by 1 and $\sdp (G_t,\{\bar u\})$ by $\|\bar u - \bar v\|^2/2 = 1 -\langle \bar u, \bar v \rangle \leq 1$ (here
we use the SDP constraint $\langle \bar u, \bar v \rangle \geq 0$). Similarly, when we add an edge between
two vertices $u$ and $v$ from the same cluster, we do not change $\plant(G_t)$, but increase $\sdp (G_t,\{\bar u\})$ by $\|\bar u - \bar v\|^2/2\geq 0$. When we add or remove a non-monotone edge, however, the value of $\plant (G_t) - \sdp (G_t,\{\bar u\})$ may increase by~1. Hence,
\begin{multline*}
\plant (G) - \sdp (G,\{\bar u\})\leq \plant (G_{sb}) - \sdp (G_{sb},\{\bar u\}) + \varepsilon m\leq\\\leq
\frac{-(a-b)\alpha N + \constsparse \sqrt{a+b(k-1)}Ns+2\varepsilon m}{2}.
\end{multline*}
This completes the proof.
\end{proof}
For $\eta\in(0,1/2]$ and $s= \sqrt{\eta (a+b(k-1))}$, we get the following corollary.
\begin{corollary}\label{cor:distances}
Under conditions of Theorem~\ref{thm:distances}, for some absolute constant $\constsparseB$, and any $\eta\in[1/(a+b(k-1)),1/2]$
\begin{equation}\label{eq:avgdisthc}
\Pr \Big(\alpha \le \frac{\left(a+b(k-1)\right)(\epsilon + \constsparseB\sqrt{\eta})}{a-b}\Big) \geq 1-2e^{-\eta m}.
\end{equation}
\end{corollary}


\newcommand{\calC}{{\cal C}}
\newcommand{\calS}{{\cal S}}
\section{First Algorithm}\label{sec:clustering}
In this section, we present our first algorithm for a partial recovery. The algorithm given the SDP solution finds a partition $V_1,\dots, V_k$ of $V$,
which is close to the planted partition $V_1^*,\dots, V_k^*$.

\begin{definition}
Consider a feasible SDP solution $\{\bar u\}_{u\in V}$. We define the center $\bar W_i$ of cluster $V_i^*$ as
$$\bar W_i = \aavg_{u\in V_i^*} \bar u.$$
For every vertex $u$ let $R_u = \|\bar u - \bar W_i\|$, where $W_i$ is the center of the cluster $V_i^*$ that contains $u$.
Let $\alpha_i = \frac{1}{2}\aavg_{u,v \in V_i^*} \|\bar u - \bar v\|^2$.
\end{definition}
\begin{definition}
Let $\rho = 1/5$ and $\Delta = 6\rho = 6/5$.  We define the core of cluster $V_i^*$ as
$$\core(i) = \{u\in V_i^*:\|\bar u - \bar W_i\| < \rho\}.$$
We say that centers $\bar W_i$ and $\bar W_j$ are well-separated if $\|\bar W_i - \bar W_j\| \geq \Delta$.
A set of clusters $\calS$ is well-separated if centers of every two clusters in $\calS$ are well-separated.
\end{definition}
We show that most clusters $V_i^*$ are well separated. First, we
establish some basic properties of centers $\bar W_i$ and parameters $\alpha, \alpha_i, \beta$.
\begin{lemma}\label{lem:basicfacts}
We have:
(1) $\aavg_i \alpha_i = \alpha$;
(2) $\aavg_{u\in V_i^*} R_u^2= \alpha_i$;
(3) $\aavg_{u\in V} R_u^2 = \alpha$;
(4) $\aavg_{i\neq j} \langle W_i, W_j\rangle = 1 - \beta = \alpha / (k-1)$;
(5) $\|\bar W_i\|^2 = 1 - \alpha_i$.
\end{lemma}
\begin{proof}
1. This follows immediately from the definitions of $\alpha$ and $\alpha_i$.

2. Write,
\begin{align*}
2\alpha_i &= \aavg_{u,v\in V_i^*} \|\bar u - \bar v\|^2 = \aavg_{u,v\in V_i^*} \|(\bar u - \bar W_i) - (\bar v - \bar W_i)\|^2\\
&= \aavg_{u,v\in V_i^*} (\|\bar u - \bar W_i\|^2 + \|\bar v - \bar W_i\|^2) -2 \aavg_{u,v\in V_i^*}\langle \bar u - \bar W_i, \bar v - \bar W_i\rangle \\
&=  2\aavg_{u\in V_i^*} R_u^2 + 0 = 2 \aavg_{u\in V_i^*} R_u^2.
\end{align*}

3. This follows from items 1 and 2.

4. Write,
\begin{align*}
\beta &= \frac{1}{2}\aavg_{i\neq j} \aavg_{u\in V_i^*, v\in V_j^*} \|\bar u - \bar v\|^2 =
\aavg_{i\neq j} \aavg_{u\in V_i^*, v\in V_j^*}  (1 - \langle \bar u, \bar  v\rangle) \\
&=1 - \aavg_{i\neq j} \Big( \langle \aavg_{u\in V_i^*} \bar u, \aavg_{v\in V_j^*} \bar v\rangle\Big)
= 1 -  \aavg_{i\neq j} \langle \bar W_i, \bar W_j\rangle.
\end{align*}
We get that $\aavg_{i\neq j} \langle \bar W_i, \bar W_j\rangle = 1 - \beta = \alpha/(k-1)$.

5. Write,
$$\alpha_i = \aavg_{u\in V_i^*} R_u^2= \aavg_{u\in V_i^*} \|\bar W_i - \bar u\|^2 =  \|\bar W_i\|^2 + 1 - 2 \aavg_{u\in V_i^*} \langle \bar W_i , \bar u \rangle
= 1 - \|\bar W_i\|^2.$$
The claim follows.
\end{proof}

\begin{lemma}\label{lem:remote-vertices}
Let $V' = \bigcup_i V_i^*\setminus \core(i)$. That is, a vertex $u$ lies in $V'$ if it does not lie in the core of the cluster that contains it.
Then
$$|V'| \leq \frac{\alpha}{\rho^2}\, kn.$$
\end{lemma}
\begin{proof}
Note that $u\in V'$ if and only if $R_u \geq \rho$, or, equivalently, $R_u^2 \geq \rho^2$. Since $\aavg_{u\in V} R_u^2 = \alpha$, we get by the Markov inequality that
$|V'| \leq (\alpha/\rho^2) |V| = (\alpha/\rho^2) kn$.
\end{proof}

We now prove that by removing at most a $\delta$ fraction of all clusters, we can obtain a well-separated set of clusters.
\begin{lemma} \label{lem:well-separated-clusters}
Let $\delta = 6\alpha / (2 - \Delta^2)$. There exists a set $\calS\subset \{V_1^*,\dots, V_k^*\}$ of well-separated clusters of size at least $(1-\delta)k$.
\end{lemma}
\begin{proof}
Let $\mu = \alpha/\delta$. From the Markov inequality and item 4 in Lemma~\ref{lem:basicfacts}, we get that there are at most
$$\frac{\alpha}{(k-1)\mu} \times \frac{k(k-1)}{2} = \frac{\alpha  k}{2\mu} = \frac{\delta k}{2} $$
unordered pairs $\{i, j\}$ with $\langle W_i, W_j\rangle \geq  \mu$. We choose one of the elements in each pair and remove the corresponding clusters.
 We obtain a set of clusters $\calS_0$ of size at least $(1 -\delta/2) k$. By the construction, for every distinct $V_i^*$ and $V_j^*$ in $\calS_0$, we have
 $\langle W_i, W_j\rangle <  \mu$.

 Let $\calS_1$ be the set of clusters $V_i^*$ with $\alpha_i \leq 2\alpha /\delta$. By the Markov inequality and item 1 in Lemma~\ref{lem:basicfacts},
 the set $\calS_1$ contains at least  $(1-\delta/2)k$ clusters.

 Finally, let $\calS = \calS_0 \cap \calS_1$. Clearly, $|\calS| \geq (1-\delta)k$. For every two clusters $V_i^*$ and $V_j^*$ in $\calS$, we have
 $$
 \|\bar W_i - \bar W_j\|^2 = \|\bar W_i\|^2 + \|\bar W_j\|^2 - 2 \langle \bar W_i, \bar W_j\rangle > (1 - \alpha_i) + (1 - \alpha_j) -2 \mu
 \geq 2(1 - 3\alpha /\delta) = \Delta^2.
 $$
\end{proof}

Now we are ready to present a greedy algorithm that finds a partition close to the planted partition.
The algorithm resembles the clustering algorithms by~\cite{CKMN01} and by~\cite{MMVCC}.

\OpenFrame
\begin{tabbing}
\quad\=\quad\=\kill
\textbf{Recovery Algorithm}\\[0.15em]
 \textbf{Input:} an optimal SDP solution $\Set{\bar u}_{u\in V}$.\\[0.15em]
 \textbf{Output:} partition $V_1,\dots, V_{k'}$ of $V$ into $k'$ clusters ($k'$ might not be equal to $k$).\\[0.3em]
\> $i = 1$; $\rho = 0.27$ \\[0.15em]
\> Define an auxiliary graph $G_{aux} = (V, E_{aux})$ with $E_{aux} = \Set{(u,v): \|\bar u - \bar v\| < 2\rho}$\\[0.15em]
\> (note that, $(u,u) \in E_{aux}$ for every $u\in V$)\\[0.15em]
\>\textbf{while} $V\setminus (V_1 \cup \dots V_{i-1}) \neq \varnothing$ \\[0.15em] 
\>\> Let $u$ be the vertex of maximum degree in $G_{aux}[V\setminus (V_1 \cup \dots V_{i-1})]$. \\[0.15em]
\>\> Let $V_i = \Set{v\notin V_1\cup \dots \cup V_{i-1}: (u,v) \in E_{aux}}$   \\[0.15em]
\>\> If $|V_i| > n$, remove $|V_i| - n$ vertices from $|V_i|$ arbitrarily, so that $|V_i| = n$.\\[0.15em]
\>\>$i = i+1$ \\[0.15em]
\>\textbf{return} clusters $V_1,\dots, V_{i-1}$.
\end{tabbing}
\CloseFrame

We will show now that the algorithm finds a ``good'' partition $V_1,\dots, V_k$. However, the clusters $V_1,\dots, V_k$ are not necessarily all of the same size.
So we cannot say that the partition is $\delta$-close to the planted partition according to Definition~\ref{def:close-and-recovery}.
We will be able, however, to prove that the partition is $\delta$-close to the planted partition in the weak sense.
\begin{definition}[cf.~with Definition~\ref{def:close-and-recovery}]
We say that a partition $V_1, \dots, V_{k'}$ is $\delta$-close to the planted partition
$V_1^*,\dots, V_k^*$ in the weak sense, if each cluster $V_i$ has size at most $n$ and there is a partial matching $\sigma$ between $1,\dots, k$ and $1, \dots, {k'}$ such that
$$\Bigl|\bigcup_{j = \sigma(i)}  V_i^* \cap V_j\Bigr| \geq (1 - \delta) kn$$
(the union is over all $i$ such that $\sigma(i)$ is defined).

We  say that $V_1, \dots, V_{k}$ is $\delta$-close to 
$V_1^*,\dots, V_k^*$ in the strong sense, if  it is $\delta$-close according to Definition~\ref{def:close-and-recovery}.
\end{definition}

\begin{theorem}\label{thm:recovery-algorithm}
The Recovery Algorithm finds a partitioning $V_1,\dots, V_{k'}$ of $V$  that is $(72\alpha)$-close to the planted partition in the weak sense.
\end{theorem}
\begin{proof}
Let $\calS$ be the set of clusters from Lemma~\ref{lem:well-separated-clusters}.
Consider a cluster $V_j$. We first show that it cannot intersect the cores of two distinct clusters $V_{i_1}^*\in\calS$ and $V_{i_2}^*\in\calS$.
Assume to the contrary that it does. Let $u_1$ be a vertex in $\core(i_1) \cap V_j$, and $u_2$ be a vertex in $\core(i_2) \cap V_j$.
Then $\|\bar W_{i_1} - \bar u_1 \| < \rho$ and $\|\bar W_{i_2} - \bar u_2 \| < \rho$. Since  $u_1,u_2\in V_j$,  vertices $u_1$ and $u_2$
have a common neighbor $u$ in the auxiliary graph $G_{aux} = (V, E_{aux})$, and, therefore, $\|\bar u_1 -\bar u_2\| < 4 \rho$.
We get that
$$\|\bar W_{i_1} - \bar W_{i_2}\| \leq \|\bar W_{i_1} - \bar u_1 \|  + \|\bar W_{i_2} - \bar u_2 \|  + \|\bar u_1 - \bar u_2 \|  < 6\rho = \Delta,$$
which is impossible since $\calS$ is a well separated set of clusters.

We now construct a partial matching $\sigma$ between clusters $V_i^*$ and $V_j$. We match every cluster $V_i^*\in \calS$
with the first cluster $V_j$ that intersects $\core(i)$ (then we let $\sigma(i) = j$).
Since each vertex belongs to some $V_j$, we necessarily match every $V_i^*\in \calS$  with some $V_j$. Moreover,
we cannot match distinct clusters $V_{i_1}^*$ and $V_{i_2}^*$ with the same $V_j$ because $V_j$ cannot intersect both cores
$\core(i_1)$ and $\core(i_2)$.

Let $Y = \bigcup_{V_i^*\in\calS} \core(i)$ and $Z = V \setminus Y$.
By Lemmas~\ref{lem:remote-vertices} and~\ref{lem:well-separated-clusters},
$$|Z|\leq \Bigl|\bigcup_i V_i^* \setminus \core(i)\Bigr| + \Bigl|\bigcup_{V_i^*\notin\calS} V_i^*\Bigr| \leq \left(\frac{1}{\rho^2} + \frac{6}{2-\Delta^2}\right) \alpha kn < 36 \alpha kn.$$

Consider a cluster $V_i^*$ and the matching cluster $V_j$.
As we proved, $V_j$ does not intersect $\core(i')$ of any $V_{i'}\in\calS$ other than $V_i$.
Therefore, $V_j \subset \core(i) \cup Z$.
We now show that
$$|V_i^* \cap V_j| \geq |\core(i)| - |Z \cap V_j|.$$
Observe that every two vertices $v_1, v_2 \in \core(i)$ are connected with an edge in $E_{aux}$ since
 $$\|\bar v_1 - \bar v_2\| \leq \|\bar v_1 - \bar W_i\|  + \|\bar v_2 - \bar W_i\| < 2 \rho.$$
In particular, every vertex $v\in \core(i)$ has degree at least $|\core(i)|$ in $G_{aux}[V\setminus (V_1 \cup \dots V_{j-1})]$.
Let $u$ be the vertex that we chose in iteration $j$. Since $u$ is a vertex of maximum degree in $G_{aux}[V\setminus (V_1 \cup \dots V_{j-1})]$,
it must have degree at least $|\core(i)|$. Now, either $V_i$ consists of all neighbors of $u$ in $G_{aux}[V\setminus (V_1 \cup \dots V_{j-1})]$ then $|V_j| \geq |\core(i)|$, or we removed some vertices from $V_j$ because it contained more than $n$ vertices, then  $|V_j| = n \geq |\core(i)|$. In either case, $|V_j| \geq |\core(i)|$.  We have,
$$
|V_i^* \cap V_j| \geq |\core(i) \cap V_j| = |V_j| -  |V_j\setminus \core(i)| = |V_j| - |V_j \cap Z|
\geq |\core(i)|- |V_j \cap Z|.
$$

Finally, using that all sets $V_j \cap Z$ are disjoint, we get
\begin{align*}
\sum_{j = \sigma(i)} |V_i^* \cap V_j| \geq \Bigl(\sum_{V_i^* \in \calS} \core(i)\Bigr) - |Z| = |Y| - |Z| =  |V| - 2 |Z| \geq (1 - 72\alpha) kn.
\end{align*}
\end{proof}

\begin{lemma}\label{lem:transform-weak-to-strong}
There is a linear-time algorithm that given a partition $V_1,\dots, V_{k'}$ of $V$  that is $\delta$-close to the planted partition in the weak sense,
outputs a partition $V_1',\dots, V_k'$ that is $(2\delta)$-close to the planted partition in the strong sense.
\end{lemma}
\begin{proof}
By the definition of the weak $\delta$-closeness, every set $V_i$ has size at most $n$. Therefore, $k' \geq k$.
We choose $k$ largest clusters among $V_1,\dots, V_{k'}$. Let $V_1',\dots, V_k'$ be these clusters. We distribute, in an arbitrary way, all vertices from other
clusters between $V_1',\dots, V_k'$ so that each of the clusters $V_i'$ contains exactly $n$ vertices.

We now show that partition $V_1',\dots, V_k'$ is $(2\delta)$-close to the planted partition in the strong sense. We may assume without loss of generality that
we chose clusters $V_1,\dots, V_k$ and that $V_i'$ consists of $V_i$ and some vertices from clusters $V_j$ with $j > k$.

Let $\sigma$ be the partial matching between clusters $V_i^*$ and $V_j$ (from the definition of the $\delta$-closeness).
We first let $\sigma'(i) = \sigma(i)$ if $\sigma(i)$ is defined and $\sigma(i) \leq k$. We get a partially defined permutation on
$\{1,\dots, k\}$.
Then we extend $\sigma$ to a permutation defined everywhere in an arbitrary way.
Write,
\begin{align*}
\Bigl|\bigcup_{j = \sigma'(i)} V_i^* \cap V_j'\Bigr| &\geq \Bigl|\bigcup_{j = \sigma(i) \leq k}  V_i^* \cap V_j \Bigr|
=
\Bigl|\bigcup_{j = \sigma(i) }  V_i^* \cap V_j\Bigr| - \Bigl|\bigcup_{j = \sigma(i) \in \{k+1,\dots, k'\}} V_i^* \cap V_j\Bigr| \\
&\geq (1 - \delta) kn - \Bigl|\bigcup_{j = \sigma(i) \in \{k+1,\dots, k'\}} V_j\Bigr|.
\end{align*}
Let
\begin{align*}
J_1 &= \{j\in \{k+1,\dots, k'\}: j =\sigma(i)\}\\
J_2 &= \{j\in \{1,\dots, k\}: j \neq \sigma(i) \text{ for every } i\}.
\end{align*}

Since $\sigma$ takes at most $k$ values, $|J_1| \leq |J_2|$. Also, $|V_{j_1}| \leq |V_{j_2}|$ for every $j_1\in J_1$ and $j_2\in J_2$ by our choice of $V_1,\dots, V_k$.
Therefore,
$$
\Bigl|\bigcup_{j \in J_1} V_j\Bigr| \leq \Bigl|\bigcup_{j \in J_2} V_j\Bigr| \leq \Bigl|\bigcup_{V_j \text{ is not matched}} V_j\Bigr| \leq \delta nk.$$
We conclude that
$$\Bigl|\bigcup_{j = \sigma'(i)} V_i^* \cap V_j'\Bigr| \geq (1 - 2\delta) nk.$$
\end{proof}

Now we are ready to prove Theorem~\ref{thm:main}.
\begin{proof}[Proof of Theorem~\ref{thm:main}]
We solve the SDP relaxation. Consider the parameter $\alpha$, which is defined by (\ref{eq:alpha}).
From Theorem~\ref{thm:distances}, we get
that $\alpha$ satisfies bounds (\ref{eq:avgdist}) with $s=2$ and (\ref{eq:avgdisthc}) with probabilities at least
$1-2 \exp(-2N)$ and $1-2\exp( -\eta m)$.
Now we run the Recovery Algorithm and find a partition $(V_1, \dots, V_k)$. By Theorem~\ref{thm:recovery-algorithm},  it is $(72\alpha)$-close to the planted partition in the weak sense. Finally, using the algorithm from Lemma~\ref{lem:transform-weak-to-strong}, we transform this partition to the desired partition $V_1',\dots, V_k'$, which 
is $(144\alpha)$-close to the planted partition.
\end{proof}

\section{Second Algorithm}\label{sec:boosting}


In this section, we present our second algorithm and prove Theorem~\ref{thm:boosting}. Theorem~\ref{thm:main-amplified} follows immediately from Theorem~\ref{thm:boosting}.

\begin{theorem}\label{thm:boosting}
Suppose that there is a polynomial-time algorithm $\cal A$ that given an instance of SBM($n$, $k$, $a/2$, $b/2$) with $\varepsilon m$ outliers finds a partition  $V_1$, \dots, $V_k$ that is $1/(10k)$-close to the planted partition (in the strong sense) with probability at least $1 - \tau$.
Then there is a randomized polynomial-time algorithm that given
an instance of SBM($n$, $k$, $a$, $b$) with $\varepsilon m$ outliers finds a partition $U_1,\dots, U_k$ that satisfies the following property.
 For every  $\delta_0 \geq ke^{-\frac{(a-b)^2}{100 a}}$, the partition $U_1,\dots, U_k$ is 
  $\delta$-close to the planted partition (in the strong sense), where
$$\delta = 4\delta_0 + \frac{80\varepsilon m}{(a-b)kn}\, ,$$
with probability at least $1 - \tau -\exp(-\delta_0 kn/6)$.
\end{theorem}
\begin{proof}
Recall that in the stochastic-block model with outliers we generate the set of edges $E$ in two steps. First, we generate
a random set of edges $E' = E_{sb}$. Then, the adversary adds and removes some edges from $E'$, and we obtain the set of edges $E$.

Let us partition all edges in $E'$ and $E$ into two groups. To this end, independently color all edges of $E\cup E'$ in two colors $1$ and $2$ uniformly at random.
Let $E_1$ and $E_2$ be the subsets of edges in $E$ colored in 1 and 2, respectively; similarly, let $E_1'$ and $E_2'$ be the subsets of edges in $E'$  colored in 1 and 2.
Denote $E^{\Delta}_i = E_i \symdiff E_i'$ for $i\in\{1,2\}$.
Note that $(V, E_1)$ is an instance of SBM($n$, $k$, $a/2$, $b/2$) with $\varepsilon m$ outliers.

Given the graph $G=(V, E)$, we generate sets of edges $E_1$ and $E_2$ (it is important that to do so, we do not have to know $E'$). We first use edges in $E_1$
to find a partition that is $1/(10 k)$-close to the planted partition. To this end,
we run algorithm $\cal A$ on $(V, E_1)$ and obtain a partition $V_1,\dots,V_k$ of $V$.


Now we use edges from $E_2$ to find a partition that is $\delta$-close to the planted partition.
We do this in two steps. First, we define a partition $U_1^0,\dots,U_k^0$,
which is close to the planted partition but not necessarily balanced -- some sets $U_i$ may contain more then $n$ vertices.
Then we transform $U_1^0,\dots,U_k^0$ to a balanced partition $U_1,\dots,U_k$.

Let us start with defining the partition $U_1^0,\dots,U_k^0$.
For technical reasons (to ensure that certain events that we consider below are independent),
it will be convenient to us to partition each set $V_i$ into two sets $V_i^L$ and $V_i^R$ containing $n/2$ vertices each
(we assume that $n$ is even; otherwise we can take sets of sizes $(n-1)/2$ and $(n+1)/2$). Denote $n' = |V_i^L| = |V_i^R| = n/2$.
Let  $V^L = \bigcup_i V_i^L$ and $V^R = \bigcup_i V_i^R$.

For every vertex $u\in V^L$, we count the number of its neighbors w.r.t. edges in $E_2$ in each of the sets
$V_1^R,\dots, V_k^R$. We find the set $V_i^R$ that has most neighbors of $u$ and add $u$ to $U_i^0$ (we break ties arbitrarily).
Similarly, for every vertex $u\in V^R$, we count the number of its neighbors w.r.t. edges in $E_2$ in each of the sets
$V_1^L,\dots, V_k^L$,  find the set $V_i^L$ that has most neighbors of $u$, and add $u$ to $U_i^0$.
We obtain a partition $U_1^0, \dots, U_k^0$.

Now we make sure that all clusters have the same size. To this end, we redistribute vertices from clusters of size greater than $n$
 among other clusters so that each cluster has size $n$. Formally, we first let $U_i^1 = U_i^0$ if $|U_i^0| \leq n$,
and let $U_i^1$ be an arbitrary subset of $n$ vertices of $U_i^0$ if $|U_i^0| > n$.
Then we arbitrarily assign all remaining vertices (i.e., vertices from $\bigcup_i U_i^0\setminus U_i^1$) among all clusters so that each cluster contains exactly $n$ vertices.
We obtain a partition $U_1,\dots, U_k$.

Let us analyze this algorithm. We may assume without loss of generality that the matching between the partition $V_1,\dots, V_k$ and the planted
partition is given by the identity permutation. Then
$$\sum_{i=1}^k |V_i^* \cap V_i| \geq nk (1-1/(10 k)) = nk - n/10.$$
In particular, for every cluster $V_i$, we have
\begin{align*}
|V_i \cap V_i^*| &\geq  9n/10,\\
|V_i \cap V_j^*| &\leq  n/10 \qquad\text{for every } j\neq i.
\end{align*}
Also, for every set $V_i^R$ (and similarly for every set $V_i^L$), we have
\begin{align}
|V_i^R \cap V_i^*| &\geq  9n/10 - n/2 = 4n'/5 , \label{ineq:good}\\
|V_i^R \cap V_j^*| &\leq  n/10 = n'/5 \qquad\text{for every } j\neq i. \label{ineq:bad}
\end{align}
Let us say that a vertex $u$ is \textit{corrupted} if it is incident on at least $T = (a-b)/20$ edges in $E_2^{\Delta}$.
\begin{claim}\label{claim:corrupted}
The total number of corrupted edges is at most $2\varepsilon m / T$.
\end{claim}
\begin{proof}
Each edge in $ E_2^{\Delta}$ is incident to at most two corrupted vertices. The total number of edges in $ E_2^{\Delta}$ is at most $\varepsilon m$.
Therefore, the number of corrupted vertices is at most $2\varepsilon m / T$.
\end{proof}

 Consider a vertex $u\in V_i^*$.  Assume that it is not corrupted. We are going to show that $u \in U_i^0$ with probability at least $1 - k e^{- \frac{(a-b)^2}{100 a}}$.

 We assume without loss of generality that $u\in V^L$.
 Let random variable $Z_j$ be the number of neighbors of $u$ in $V_j^R$ w.r.t. edges in $E_2'$.
 Consider the event ${\cal E}_u$ that $Z_i \geq (3a + 2b)/20$ and $Z_j \leq (2a + 3b)/20$ for every $j\neq i$.
 We will prove now that if ${\cal E}_u$ happens then $u\in U_i^0$. After that we will show that the probability that ${\cal E}_u$ does not happen is exponentially small.

 Assume to the contrary that ${\cal E}_u$ happens but $u\in U_j^0$ for some $j\neq i$. Then $u$ has at least as many neighbors in $V_j^R$ as is in $V_i^R$.
 Let $A_+$ be the number of edges $e\in E_2\setminus E_2'$  from $u$ to vertices in $V_j$ (edges added by the adversary); and
  $A_-$ be the number of edges in $e\in E_2'\setminus E_2$ from $u$ to $V_i$ (edges removed by the adversary).
  Then $A_+ + A_- < T$ since $u$ is not corrupted.
  Observe that there are at most $Z_j + A_+$ edges $e\in E_2$ from $u$ to $V_j^R$; there are at least $Z_i - A_-$ edges $e\in E_2$ from $u$ to $V_i^R$.
  Therefore,
  $Z_j + A_+ \geq Z_i - A_-$, and hence (using that event ${\cal E}_u$ happens)
  $$T > A_+ + A_- \geq Z_i - Z_j \geq \frac{3a +2b}{20} - \frac{2a +3b}{20} = \frac{a - b}{20} =T,$$
  we get a contradiction.

  We use the Bernstein inequality to upper bound the probability that ${\cal E}_u$ does not happen.
  Note that for every $j$ (including $j=i$), $u$ is connected to every vertex in $V_j^R \cap V_i^*$ by an edge in $E_2'$ with
  probability $a/(2n)$; $u$ is connected to every vertex in $ V_j^R \setminus V_i^* $ by an edge in $E_2'$ with
  probability $b/(2n)$. Using bound~(\ref{ineq:good}), we get that the expected number
  of neighbors of $u$ in $V_i^R$  is at least $(4a + b)/10$. That is, $\E[Z_i] \geq (4a + b)/10$. By the Bernstein inequality,
  $$\Pr[Z_i < (3a + 2b)/10] \leq e^{- \frac{(a-b)^2/100}{2((4a + b)/10 + (a-b)/30)}}  \leq e^{- \frac{(a-b)^2}{100a}}.$$
  Similarly, using bound~(\ref{ineq:bad}), we get that for every $j\neq i$,  $\E[Z_j] \leq (a + 4b)/5$. By the Bernstein inequality,
  $$\Pr[Z_j > (2a + 3b)/10] \leq e^{- \frac{(a-b)^2/100}{2((a + 4b)/10 + (a-b)/30)}}  \leq e^{- \frac{(a-b)^2}{100a}}.$$
By the union bound, $\Pr[{{\cal E}_u}] \geq 1 - k e^{- \frac{(a-b)^2}{100a}}$.

 \medskip

 We proved that for every $u\in V_i^*$, $\Pr[u\in U^0_i] \geq 1- \delta_0$ (recall that $\delta_0  \geq k e^{- \frac{(a-b)^2}{100a}}$).
 Let $B_L$  be the number of vertices in $V_L$ such that ${\cal E}_u$ does not happen, and
 $B_R$  be the number of vertices in $V_R$ such that ${\cal E}_u$ does not happen.

 Note that $\E[B_L] \leq \delta_0 kn'$. Also all events ${\cal E}_u$ with $u\in V_L$ are independent since each event ${\cal E}_u$ depends only on the subset of edges
 of $E_2$  that  goes from $u$ to $V_R$. Therefore, by the Chernoff bound
 $$
 \Pr[B_L \geq 2\delta_0 kn']  < e^{-\delta_0 kn'/3} = e^{-\delta_0 kn/6}.
 $$
 Similarly, $\Pr[B_R \geq 2\delta_0 kn']  < e^{-\delta_0 kn/6}$, and $\Pr[B_L + B_R \geq 2\delta_0 kn] < 2 e^{-\delta_0 kn/6}$.

 Assume now that $\Pr[B_L + B_R < 2\delta_0 kn]$. Then
 $$\E\Bigl[\sum_{i=1}^k |V_i^* \cap U_i^0|\Bigr] \geq (1 -2\delta_0) kn - 40\varepsilon m/(a-b) = (1 - \delta/2) kn.$$
 here, $40\varepsilon m/(a-b)$ is the upper bound on the number of corrupted vertices from Claim~\ref{claim:corrupted}, and $\delta = 4\delta_0 + \frac{80\varepsilon m}{(a-b)kn}$ as in the statement of the theorem.
 Now,
$$\sum_{i=1}^k |V_i^* \cap U_i| \geq \sum_{i=1}^k |V_i^* \cap U_i^1|\geq
\sum_{i=1}^k |V_i^* \cap U_i^0| - \sum_{i:|U_i^0| > n} (|U_i^0| - n) \geq (1 - \delta/2) kn - (\delta/2) kn = (1 -\delta) kn.$$
We proved that $U_1,\dots, U_k$ is  $\delta$-close to the planted partition, when algorithm $\cal A$ succeeds and $B_L + B_R < 2\delta_0 kn$; that is,
with probability at least $1 - \tau - \exp(-\delta_0 kn/6)$.
\end{proof}

Now we present the proof of Theorem~\ref{thm:main-amplified}.
\begin{proof}[Proof of Theorem~\ref{thm:main-amplified}]
Observe that under our assumption that
$$\frac{ \sqrt{a+b(k-1)} }{a-b}+\frac{\epsilon\left(a+b(k-1)\right)}{a-b} \leq c/k,$$
our first algorithm finds a partition that is $1/(10k)$-close to the planted partition given an instance of $\SBM(n, k, a/2, b/2)$.
Hence, we can apply Theorem~\ref{thm:boosting} and get a partition that is $\delta$-close to the planted partition, as desired.
\end{proof}

\section{KL-divergence}\label{sec:KL}
\begin{proof}[Proof of Theorem~\ref{thm:KL}]
Theorem~\ref{thm:KL} almost immediately follows from Theorem~\ref{thm:main} and Lemma~\ref{lem:klbounds}.
\begin{lemma} \label{lem:klbounds}
Consider two distributions $P, Q$ over the same sample space $\Omega$. For every event $\calE\subset \Omega$, we have
\begin{equation}
 Q(\calE) \leq \max\Big(\frac{2\dkl{Q,P}}{-\log P(\calE) + 1}, e\sqrt{2 P(\calE)}\Big),
\end{equation}
where $d_{KL}(Q,P)$ is the Kullback--Leibler divergence of $P$ from $Q$.
\end{lemma}

Consider the worst adversary $A$ for
the algorithm from Theorem~\ref{thm:main} --- that is, the adversary for which the algorithm succeeds to recover a $(1-\delta)$
fraction of vertices with the smallest probability. The adversary takes the graph $G\sim {\cal G}$
and transforms it to $A(G)$. Without loss of generality we may assume that the adversary is deterministic.
Let $\calE$ be the set of graphs $G$ for which our algorithm fails to recover $\delta$ fraction of vertices on the corrupted graph $A(G)$. By Theorem~\ref{thm:main}, the probability of $\calE$ in the Stochastic Block Model distribution is at most
$2\exp( -\eta m)$.
Thus, by Lemma~\ref{lem:klbounds}, the probability of $\calE$ in the distribution of ${\cal G}$  is 
bounded as
$$\delta \leq \max\Big(\frac{2\lambda m}{\eta m}, 2 e^{\frac{-\eta m}{2}+1}\Big) = 
\max\Big(\frac{2\lambda}{\eta}, 2 e^{\frac{-\eta m}{2}+1}\Big).$$
\end{proof}

We now prove an auxiliary Lemma~\ref{lem:data-proc} and then Lemma~\ref{lem:klbounds}.
\begin{lemma}\label{lem:data-proc}
Consider two distributions $P, Q$ over the same sample space $\Omega$. Suppose that $\Omega$ is the union of disjoint
events $\calE_i$. Then
$$\dkl{Q,P} \geq \sum_{i} Q(\calE_i) \log \frac{Q(\calE_i)}{P(\calE_i)}.$$
\end{lemma}
\begin{proof}
By the definition, KL divergence equals
$$\dkl{Q,P} = \sum_{\sigma \in \Omega} Q(\sigma) \log\frac{Q(\sigma)}{P(\sigma)}  =
\sum_i \sum_{\sigma \in {\cal E}_i} Q(\sigma) \log\frac{Q(\sigma)}{P(\sigma)} .$$
We lower bound each of the terms on the right hand side using the log-sum inequality (which follows from
the convexity of the function $x\mapsto x\log x$ and Jensen's inequality).
\begin{claim}[Log-sum inequality see e.g.~\cite{CKBook}] Let $q_1,\dots, q_T$ and  $p_1,\dots, p_T$ be nonnegative numbers. Then,
$$\sum_i q_i \log\frac{q_i}{p_i} \geq \Big(\sum_i q_i \Big)
\log \Big(\frac{\sum_i q_i}{\sum_i p_i} \Big).$$
\end{claim}
We get
$$\dkl{Q,P}\geq \sum_i \Big(\sum_{\sigma \in \calE_i}Q(\sigma)\Big)
\log \frac{\sum_{\sigma\in \calE_i}Q(\sigma_i)}{\sum_{\sigma\in \calE_i}P(\sigma_i)} =
\sum_i Q(\calE_i)\log \frac{Q(\calE_i)}{P(\calE_i)}.
$$
\end{proof}

\begin{proof}[Proof of Lemma~\ref{lem:klbounds}]
We apply Lemma~\ref{lem:data-proc} to the events $\calE$ and $\bar \calE = \Omega\setminus \calE$:
\begin{equation}\label{eq:dkl-sum-i}
\dkl{Q,P} \geq
Q(\calE) \log \frac{Q(\calE)}{P(\calE)}  +
Q(\bar \calE) \log \frac{Q(\bar \calE)}{P(\bar \calE)} .
\end{equation}
We bound the second term on the right hand side using the inequality $x\log x \geq (x - 1)\log e$ for $x \geq 0$:
\begin{multline*}
Q(\bar \calE)
\log\frac{Q(\bar \calE)}{P(\bar \calE)}
=
P(\bar \calE)\times \Big[\frac{Q(\bar \calE)}{P(\bar \calE)}
\log\frac{Q(\bar \calE)}{P(\bar \calE)}\Big]\geq P(\bar \calE)\times
\Big[\frac{Q(\bar \calE)}{P(\bar \calE)} - 1\Big]\log e =\\=
(Q(\bar \calE) - P(\bar \calE)) \log e =
(P(\calE) - Q(\calE)) \log e \geq - Q(\calE) \log e.
\end{multline*}
We have
$$\dkl{Q,P}\geq Q(\calE) \log\frac{Q(\calE)}{P(\calE)}  - Q(\calE) \log e 
= Q(\calE) \log\frac{Q(\calE)}{e\cdot  P(\calE)}.$$
Thus, either
$Q(\calE)\leq e\sqrt{2P(\calE)}$, or $Q(\calE)/(eP(\calE))\geq \sqrt{2/P(\calE)}$, and, consequently,
$$Q(\calE) \le \frac{2\dkl{Q,P}}{-\log (P(\calE)) + 1}.$$
\end{proof}


\acks{We would like to thank Elchanan Mossel for some useful discussions. We would also like to thank the anonymous referees for COLT'16 for pointing out a correction in Lemma 21.}
\notCOLT{\bibliographystyle{plainnat}}

\bibliography{dblp,semirandom}

\appendix

\section{Concentration Inequalities}\label{sec:appendix}
\subsection{Bernstein Inequality}

We will use the following standard inequality known as the Bernstein inequality or the Hoeffding inequality (see e.g., Theorem 2.7. in \cite{mcdiarmid1998}).
\begin{fact}\label{fact:bernstein}
Let $X_1,\dots, X_n$ be independent random variables with $X_i - \E[X_i] \leq B$ for all $i$. Then
\begin{equation}
\Pr\big[\sum_i X_i  - \E \sum_i X_i> t  \big] \le \exp\left(- \frac{t^2/2}{\sigma^2+ B t/3} \right),
\end{equation}
where $\sigma^2 = \sum_i \Var[X_i]$. For Bernoulli random variables taking values 0 and 1, $\sigma^2 \leq \sum_i \E X_i$.
\end{fact}

\subsection{Size of the Planted Cut}
In this section, we upper bound the size of the planted cut. For convenience, we give the same probability estimate as in Theorem~\ref{thm:fromGV}.
\begin{lemma}\label{lem:bernst-plant-cut}
For a random graph $G_{sb}$ from the Stochastic Block model $\SBM(n, k, a, b)$, we have
$$\Pr(\plant (G_{sb}) \leq b (k-1) N/2 + 2\sqrt{a+b(k-1)}Ns) \geq 1 - e^{-\frac{9 s^2}{4 + 8s/\sqrt{a+b(k-1)}}\,N}.$$
\end{lemma}
\begin{proof} The expected size of the planted cut is $b (k-1) N/2$. Thus, by the Bernstein inequality, we have
$$\Pr(\plant (G_{sb}) \leq b (k-1) N/2 + t) \leq e^{-\frac{t^2}{b (k-1) N + 2t/3}}.$$
For $t = 2\sqrt{a+b(k-1)}Ns$, we get
$$\Pr(\plant (G_{sb}) < b (k-1) N + t) \leq e^{-\frac{s^2N}{1/4  + s/(3\sqrt{a+b(k-1)})}} <
e^{-\frac{9 s^2}{4 + 8s/\sqrt{a+b(k-1)}}\,N}.$$
\end{proof}

\subsection{Proof of Theorem~\ref{thm:fromGV}}\label{app:thm:fromGV}
In this section, we prove Theorem~\ref{thm:fromGV}, which is an analog of Lemma 4.1 in~\cite{GV14}.
The proof closely follows their proof. In the proof, we will use the Grothendieck inequality~(see~\cite{Gro,Krivine,BMMN}).

\begin{theorem}[Grothendieck inequality] For every $n\times n$ matrix $M$, the following inequality holds
\begin{equation}\label{eq:grothendieck}
\max_{\|U_i\|,\|V_j\|=1} \Big| \sum_{i,j=1}^{n} M_{ij}  \langle U_i, V_j \rangle\Big| \le \Kg \cdot \max_{x,y \in \{-1,1\}^n} \sum_{i,j=1}^n M_{ij} x_i y_j,
\end{equation}
where $K_G\leq 1.783$ is the Grothendieck constant. The first maximum is over all unit vectors $U_1,\dots, U_n$ and $V_1,\dots, V_n$.
\end{theorem}
\begin{proof}[Proof of Theorem~\ref{thm:fromGV}]
Let $L$ be the Laplacian of the graph $G_{sb}$ and $\Delta L = L - \E[L]$. For any feasible
SDP solution $\{\tilde u\}$, we have
$$\sum_{u<v} \Delta a_{uv} \| \tilde u -  \tilde v\|^2 =
\sum_{u,v} \Delta L_{uv} \langle \tilde u, \tilde v\rangle.$$
We upper bound the right hand side using the Grothendieck inequality (with $U_i=V_i = \tilde u$, where $u$ is 
the $i$-th vertex in the graph):
\begin{eqnarray}
\max_{\{\tilde u\}}\Big| \sum_{u,v} \Delta L_{uv} \langle \tilde u, \tilde v\rangle\Big|
&\leq& K_G \max_{x,y\in \{-1,1\}^n}\sum_{u,v} \Delta L_{uv} \, x_u y_v \label{eq:groth}\\
&=& K_G \max_{x,y\in \{-1,1\}^n} \sum_{u<v} \Delta a_{uv}  (x_u - x_v)(y_u - y_v)\nonumber.
\end{eqnarray}
Note that each $\Delta a_{uv}$ is a Bernoulli random variable taking values $-\E[a_{uv}]$ and $1-\E[a_{uv}]$ with
probabilities $1-\E[a_{uv}]$ and $\E[a_{uv}]$, respectively. All values $(x_u - x_v)(y_u - y_v)$ lie in the set
$\{-4,0,4\}$. By the Bernstein inequality, for fixed $x,y\in\{-1,1\}^n$,
we have
$$\Pr\Big(\sum_{u<v} \Delta a_{uv}  (x_u - x_v)(y_u - y_v)\geq t\Big) \leq e^{-\frac{t^2}{2\sigma^2 (x,y) + 8t/3}},$$
where
$$\sigma^2(x,y) = \sum_{u<v}\Var\big[\Delta a_{uv}  (x_u - x_v)(y_u - y_v)\big] = \sum_{u<v} \Var[\Delta a_{uv}](x_u - x_v)^2(y_u - y_v)^2.$$
Since the set of edges $E_{sb}$ comes
from the stochastic block model, we have $\E a_{uv}=a/n$ if $(u, v) \in (V \times V)_{in}$,
and $\E a_{uv} = b/n$ if $(u,v) \in (V \times V)_{out}$. Note that $\Var[\Delta a_{uv}]= \E[a_{uv}]\,(1-\E[a_{uv}]) < \E[a_{uv}]$.
Thus,
\begin{eqnarray*}
\sigma^2 (x,y)
&\leq& \frac{a}{n} \sum_{\substack{(u,v)\in (V \times V)_{in}\\u<v}} (x_u - x_v)^2(y_u - y_v)^2 +
\frac{b}{n} \sum_{\substack{(u,v)\in (V \times V)_{out}\\u<v}} (x_u - x_v)^2(y_u - y_v)^2\\
&\leq& \frac{4a}{n} \sum_{\substack{(u,v)\in (V \times V)_{in}\\u<v}} (x_u - x_v)^2 +
\frac{4b}{n} \sum_{\substack{(u,v)\in (V \times V)_{out}\\u<v}} (x_u - x_v)^2\\
&=& \frac{4(a-b)}{n} \sum_{\substack{(u,v)\in (V \times V)_{in}\\u<v}} (x_u - x_v)^2 +
\frac{4b}{n} \sum_{\substack{(u,v)\in V \times V\\u<v}} (x_u - x_v)^2.
\end{eqnarray*}
For any set $S\subset V$,
$$\sum_{\substack{(u,v)\in S \times S\\u<v}} (x_u - x_v)^2 = 4\, |\{u\in S: x_u =-1\}|\cdot
|\{v\in S: x_v =1\}|\leq |S|^2.$$
Hence,
$$\sigma^2 (x,y) \leq \frac{4(a-b)}{n}\times k \times n^2 + \frac{4b}{n}\times (nk)^2 =
4aN + 4b(k-1)N.$$
Consequently,
$$\Pr\Big(\sum_{u<v} \Delta a_{uv}  (x_u - x_v)(y_u - y_v)\geq t\Big)\leq e^{-\frac{t^2}{8(a+b(k-1))N + 8t/3}}.$$
Using the union bound over all $x,y\in \{-1,1\}^V$, we get
$$\Pr\Big(\max_{x,y\in\{-1,1\}^n}\sum_{u<v} \Delta a_{uv}  (x_u - x_v)(y_u - y_v)\geq t \Big)\leq 2^{2N}e^{-\frac{t^2}{8(a+b(k-1))N + 8t/3}}.$$
By~(\ref{eq:groth}),
$$\Pr\Big(\max_{\{\tilde u\}}\Big|\sum_{u<v} \Delta a_{uv}  \|\tilde u - \tilde v\|^2 \Big|\geq K_G \, t \Big)\leq 2^{2N}e^{-\frac{t^2}{8(a+b(k-1))N + 8t/3}}=
e^{-\frac{t^2}{8(a+b(k-1))N + 8t/3}+2N\ln2}.$$
Let $t= 6\sqrt{a+{b(k-1)}}Ns$. Then,
\begin{eqnarray*}
\frac{t^2}{8(a+b(k-1))N + 8t/3}-2N\ln2 &=& \Big(\frac{9 s^2}{2 + 4s/\sqrt{a+b(k-1)}}-2\ln2\Big)N\\
&\geq& \frac{9 s^2}{4 + 8s/\sqrt{a+b(k-1)}}\,N.
\end{eqnarray*}
The last inequality holds for $s\geq 1$ and $a+b(k-1)\geq 11$. Therefore,
$$\Pr\Big(\max_{\tilde u}
\Big|\sum_{u<v} \Delta a_{uv}  \|\tilde u - \tilde v\|^2\Big| \geq 6K_G \sqrt{a+b(k-1)} N s\Big)\leq
e^{-\frac{9 s^2}{4 + 8s/\sqrt{a+b(k-1)}}\,N}.$$
\end{proof}

\section{Lower Bounds}\label{sec:lb}

In this section we give lower bounds on the partial recovery in the model with two communities. We show that it is not possible to recover a $\delta$ fraction of all vertices in the pure Stochastic Block Model if
\begin{equation}\label{eq:Bound1}
(a-b) < C \sqrt{(a+b) \ln 1/\delta},
\end{equation}
for some constant $C$, and  it is not possible to recover a $\delta$ fraction
of all vertices in the Stochastic Block Model with Outliers (where the adversary is allowed to add at most $\varepsilon (a+b)n$ edges) if
\begin{equation}\label{eq:Bound2}
(a-b) < C\varepsilon \delta^{-1}(a+b).
\end{equation}

We note that very recently \cite{ZZ} showed a lower bound with a dependence similar to ~\eqref{eq:Bound1}.
For simplicity of exposition we slightly alter the Stochastic Block Model. We consider graphs with parallel edges.  The
number of edges between two vertices $u$ and $v$ in the new model is not a Bernoulli random variable with parameter $a/n$ or $b/n$ as in the standard
Stochastic Block Model, but a
Poisson random variable with parameter $a/n$ or $b/n$. Note that recovering partitions in the Poisson model with very slightly
modified parameters $a'=n \ln (1-a/n)$ and $b' = n \ln (1-b/n)$, is not harder
than in the Bernoulli model, since the algorithm may simply replace parallel edges with single edges
and obtain a graph from the standard Stochastic Block Model.

Before proceeding to the formal proofs, we  informally discuss why these bounds hold. Consider two vertices $u$ and $v$ lying in the opposite
clusters. Suppose we give the algorithm not only the graph $G$, but also the correct clustering of all vertices but $u$ and $v$. The algorithm needs now to decide
where to put $u$ and $v$. It turns out that the only \textit{useful} information the algorithm has about $u$ and $v$ are the four numbers -- the number of neighbors $u$ and $v$
have in the left and right clusters. These numbers are distributed according to the Poisson distribution with parameters $a$ and $b$. So the algorithm
is really given four numbers: two numbers $X_1, Y_1$ for vertex $u$ and two numbers $Y_2, X_2$ for vertex $v$. The algorithm needs to
decide whether
\begin{enumerate}
\item[(a)]
$X_1$ and $X_2$ have the Poisson distribution with parameter $a$, and
$Y_1$ and $Y_2$ have the Poisson distribution with parameter $b$; or
\item[(b)]
$X_1$ and $X_2$ have the Poisson distribution with parameter $b$, and
$Y_1$ and $Y_2$ have the Poisson distribution with parameter $a$.
\end{enumerate}
We show in Corollary~\ref{lem:pois-couple} that no test distinguishes (a) from (b) with error probability less than $\delta$ given
by the bound~(\ref{eq:Bound1}). This implies~(\ref{eq:Bound1}).

To prove the bound~(\ref{eq:Bound2}), we first specify what the adversary does in the model with outlier edges (noise). It
picks $\delta n$ fraction of all vertices on the left side and on the right side. For each chosen vertex, it adds approximately $(a-b)$
extra edges going to the opposite side. After that every chosen vertex has the same distribution of edges going to the opposite cluster as to its own cluster.
Hence, the chosen vertices on the left side and chosen vertices on the right side are statistically indistinguishable.
To add $(a-b)$ extra edges to every chosen vertex, the adversary needs $2(a-b)\delta n$ edges, but
he has a budget of $\Theta(\varepsilon(a+b)n)$ edges. This gives the bound~(\ref{eq:Bound2}).

In the rest of the section, we use the ideas outlined above to prove the following theorem. In the proof, we couple
the distribution of the random variables $(X_1,Y_1), (Y_2,X_2)$ with the distribution of graphs in the Stochastic Block Model.

\begin{theorem}\label{thm:lb}
It is statistically impossible to recover more than $\delta$ fraction of all vertices if the bound \ref{eq:Bound1} holds in the Stochastic Block Model,
and if the bound \ref{eq:Bound2} holds in the Stochastic Block Model with Outliers,  where the adversary can add at most $O(\varepsilon (a+b)n)$ edges.
The constant $C$ is a universal constant.
\end{theorem}

\subsection{Adversary in the SBM with Outliers}\label{sec:sr-adv}
We first describe the adversary for generating graphs in the SBM with Outlier edges. The adversary fixes two sets
$L'\subset L$ and $R'\subset R$ in the left and right clusters of size $\rho n$ each for $\rho = \Theta (\varepsilon (a+b)/(a-b))$.
Let $L'' = L\setminus L'$ and $R''=R\setminus R'$. Then it generates a graph according to the pure Stochastic Block Model.
The adversary counts the number of edges going from $L'$ to $R''$, and the number of edges
going from $R'$ to $L''$. Denote these numbers by $Z_{L'}$ and $Z_{R'}$ respectively.
Then, the adversary independently computes two numbers ${\kappa}_{L'} = \hat{\kappa}(Z_{L'})$ and
${\kappa}_{R'} = \hat{\kappa}(Z_{R'})$ using a random function  $\hat{\kappa}$ we describe in a moment.
He adds $\kappa_{L'}$ edges between $L'$ and $R''$
and  $\kappa_{R'}$ edges between $R'$ and $L''$. He adds the edges one by one
every time adding one edge between a random vertex in $L'$ and a random vertex in $R''$ or between a random vertex in $R'$ and a random vertex in $L''$.

Denote $M=\rho (1-\rho)n$. In Corollary~\ref{cor:kappa}, we show that there exists a function $\hat{\kappa}$ upper bounded by
$(a-b) M$
such that the total variation distance between $P_{1}$ and $\hat{\kappa}(P_2)$ is at most $1/2$, where $P_{1}$ and $P_{2}$
are Poisson random variables with parameters $aM$ and $bM$.
The adversary uses this function $\hat{\kappa}$. Note that he adds at most
$$ 4(a-b)M = 4(a-b)\rho (1-\rho)n\leq 4(a-b) \rho n = \Theta (\varepsilon (a+b) n),$$
edges.

\subsection{Restricted Partitioning}
Let us partition the sets $L$ and $R$ into two sets each: $L=L'\cup L''$ and $R=R'\cup R''$. We partition the sets
before we generate the graph from the Stochastic Block Model, and thus the partitioning does not depend on the edges
present in the graph. Consider the following classification task:
the classifier gets the graph $G$ generated according to the Stochastic Graph Model (with or without the adversary) and the sets
$L'$, $R'$, $L''$ and $R''$ (which were chosen before the graph was generated). We specify that $L'' \subset L$ and $R''\subset R$.
However, we swap the order of $L'$ and $R'$ with probability $1/2$. Thus the classifier does not know whether $L'\subset L$ or $L'\subset R$ and
whether $R'\subset L$ or $R'\subset R$.
Its goal is to guess whether $L'\subset L$ or $L'\subset R$ and, consequently, whether
$R'\subset L$ or $R'\subset R$. We call this classifier a restricted classifier.

\begin{lemma}[Restricted Classifier for pure Stochastic Block Model]\label{lem:restr-pure}
If there exists a procedure that recovers partitions in the pure Stochastic Block Model with accuracy at least $1- \delta$,
then there exists a restricted classifier (as above) for sets $L'=\{u\}$,  $R'=\{v\}$, $L''=L\setminus \{u\}$ and $R''=\setminus \{v\}$
that errs with probability at most $2\delta + 1/n$.
\end{lemma}
\begin{proof}
The classifier works as follows. It executes the recovery procedure for the input graph $G=(V,E)$ and gets two sets $S^*$ and $T^*$.
It picks at random $w'\in\{u,v\}$ and $w''\in L''$. Now if $w'$ and $w''$ lie in the same set $S^*$ or $T^*$, then the algorithm
returns ``$w'\in L$'', otherwise it returns ``$w'\in R$''.
What is the error probability of this classifier?

Since the distribution of graphs in the Stochastic Block Model is invariant under permutation of vertices in $L$ and in $R$, the error probability
will not change if we alter the process as follows: the classifier first runs the recovery procedure, then we pick two random vertices $u\in L$ and $v\in R$ and give these vertices
to the classifier. Note that the classifier does not need $u$ and $v$ to run the recovery procedure. Let us compute the error probability. Suppose that the recovery procedure misclassified $\delta^*$ fraction of all vertices, and say $S^*$ corresponds to $L$ i.e. $|S^*\cap L|= (1 - \delta^*)n$.  If the algorithm picks $w' = u \in L$, then the probability that $w',w''\in S^*$ equals
$(1 - \delta^*)((1 - \delta^*)n -1)/n\geq 1-2\delta^* - 1/n$.
 Similarly, if $w'=v\in R$, then
the probability that $w'\in T^*$ and $w''\in S^*$ equals
$(1 - \delta^*)^2\geq 1 -2\delta^*$.

Since the expected value of $\delta^*$ is at most $\delta$ we get the desired result.
\end{proof}

We now prove a similar lemma for Stochastic Block Model with Outlier edges.

\begin{lemma}[Restricted Classifier for Stochastic Block Model with Outliers]\label{lem:restr-sr}
If there exists a procedure that recovers partitions in the Stochastic Block Model with Outlier Edges with accuracy at least $1- \delta$,
then there exists a restricted classifier for sets $L'$,  $R'$, $L''=L\setminus L'$ and $R''=\setminus R'$
with $|L'|=|R'| < n/2$ that errs with probability at most $\delta n/|L'|$.
\end{lemma}
\begin{proof}
As before, the classifier executes the recovery procedure for the input graph $G=(V,E)$ and gets two sets $S^*$ and $T^*$.
Then, the classifier picks sets $W'\in \{L',R'\}$ and $W'' = \{L'', R''\}$ at random.  It also picks random vertices
$w'\in W'$ and $w''\in W''$. If $w'$ and $w''$ lie in the same set
 $S^*$ or $T^*$, the classifier returns ``$W'$ and $W''$ are on the same side of the cut''; otherwise,
 it returns ``$W'$ and $W''$ are on different sides of the cut''. Note that the classifier knows whether $W''=L''$ or $W''=R''$, and
 hence whether $W''$ lies on the left or right side of the cut.

Let $\delta^*$ be the fraction of misclassified vertices. Further, let $\delta'$ be the fraction of misclassified vertices in $L'\cup R'$; and
$\delta''$ be the fraction of misclassified vertices in $L''\cup R''$. Note that $\delta^* = \big(\delta' (|L'|+|R'|) + \delta'' (|L''|+|R''|)\big)/(2n)$.
The error probability of the classifier given the partition $S^*$ and $T^*$ is at most
$$1 - (1-\delta')(1-\delta'')\leq \delta_1+\delta_2\leq \frac{2\delta^*n}{|L'|+|R'|}= \frac{\delta^*n}{|L'|}.$$

The error probability over random choices of the graph is at most $\E[\delta^*n/|L'|]=\delta n/|L'|$.
\end{proof}

In the next subsection, we argue that, in a way, the only useful information the restricted classifier can use about the graph given the sets $L'$, $R'$, $L''$ and $R''$
are the number of edges between sets $L'$, $L''$, $R'$ and $R''$.

\subsection{Tests for Pairs of Distributions}

Let $D_1$ and $D_2$ be two distributions; and let $D_{Left}=D_1\times D_2$ and $D_{Right} = D_2\times D_1$ be the product
distributions -- distributions of pairs $(X,Y)$ and $(Y,X)$, where $X$ and $Y$ are independent random variables distributed
as $D_1$ and $D_2$ respectively. In this section, we consider tests that given two independent
pairs of random variables $(X_1, Y_1)$ and $(Y_2,X_2)$ distributed according to $D_{Left}$ and $D_{Right}$
needs to decide which pair is drawn from $D_{Left}$ and which from $D_{Right}$.
The test gets the pairs as an unordered set $\{(X_1, Y_1), (Y_2,X_2)\}$. We show
that the restricted classifier is essentially a test for distributions $D_1$ and $D_2$,
where $D_1$ is the distribution of the total number of edges between $L'$ and $L''$; $D_2$ is
the distribution of the number of edges between $R'$ and $R''$.

\begin{lemma}\label{lem:rest-to-test}
Consider the Block Stochastic Model with sets $L'$, $R'$, $L''$, $R''$ as in Lemma~\ref{lem:restr-pure}, or
the Stochastic Model with Outlier edges, with sets $L'$, $R'$, $L''$, $R''$ as in Lemma~\ref{lem:restr-sr}. When we have outlier edges (noise), we assume
that the adversary behaves as described in Section~\ref{sec:sr-adv} and the sets $L'$ and $R'$ he chooses are the same sets as above. Let
$D_1$ be the distribution of the number of edges between $L'$ and $L''$, and
$D_2$ be the distribution of the number of edges between $L'$ and $R''$. (Note, that the number of edges between $R'$ and $R''$ is also distributed as $D_1$;
the number of edges between $R'$ and $L''$ is distributed as $D_2$.) Then, if there exists a restricted classifier (see the previous section) with error
probability at most $\delta$, then there exists a test that decides whether
\begin{itemize}
\item $(X_1,Y_1) \sim D_1\times D_2$ and $(Y_1,X_1) \sim D_2\times D_1$; or
\item $(X_1,Y_1) \sim D_2\times D_1$ and $(Y_1,X_1) \sim D_1\times D_2$
\end{itemize}
 with error probability at most $\delta$.
\end{lemma}
\begin{proof}
Suppose we are given a restricted classifier with error probability at most $\delta$. We construct a test for pairs $D_1\times D_2$ and $D_2\times D_1$.
The test procedure receives two pairs $(X_1,Y_1)$ and $(Y_2, X_2)$. Then it generates a graph from the model (the pure SBM, or the one with outlier edges) as follows. It creates four
sets of vertices $A$, $B$, $L''$ and $R''$. It adds edges to the subgraphs on $A\cup B$ and $L''\cup R''$ as in the Stochastic Block Model with
planted cuts $(A,B)$ and $(L'', R'')$ respectively. Then, it adds $X_1$, $Y_1$, $X_2$, $Y_2$ edges
between $A$ and $L''$, $A$ and $R''$, $B$ and $R''$, $B$ and $L''$ respectively. These edges are added at random one by one: say, to add an edge
between $A$ and $L''$, the test procedure picks a random vertex in $A$ and a random vertex in $L''$ and connects these vertices with an edge.
Once the graph is generated, the procedure executes the restricted classifier. If the classifier tells that $A$ and $L''$ are on the same side of the cut, the test
returns that $X_1,X_2\sim D_1$ and $Y_1,Y_2\sim D_2$; otherwise, $X_1,X_2\sim D_2$ and $Y_1,Y_2\sim D_1$.

We now analyze the tester. We claim that the graph obtained by the procedure above is distributed according to the model
(the pure SBM, or the one with outlier edges), and the planted cut is $(A\cup L'', B\cup R'')$ if $X_1,X_2\sim D_1$ and $Y_1,Y_2\sim D_2$;
the planted cut is $(B\cup L'', A\cup R)$ if $X_1,X_2\sim D_2$ and $Y_1,Y_2\sim D_1$. For the proof, assume without loss of generality that
$X_1,X_2\sim D_1$ and $Y_1,Y_2\sim D_2$.

Let $N_{uv}$ be the number of edges between vertices $u$ and $v$. In the pure Stochastic Block Model,
we need to verify that random variables $N_{uv}$ are independent; and $N_{uv}$ has the Poisson distribution with parameter $a/n$ for $(u,v)\in A\times L''$
and $(u,v)\in B\times R''$; $N_{uv}$ has the Poisson distribution with parameter $b/n$ for $(u,v)\in A\times R''$
and $(u,v)\in B\times L''$.
 This immediately follows from the following Poisson Thinning Property,
since $X_1$, $X_2$, $Y_1$ and $Y_2$ have Poisson distributions
with parameters $(a/n) |A|\cdot |L'|$, $(a/n) |B|\cdot |R|'$, $(b/n) |A|\cdot |L'|$, $(b/n) |B|\cdot |R|'$ respectively.

\begin{fact} Suppose we pick a number $P$ according to the Poisson distribution with parameter $\lambda$. Then, we distribute $P$ balls into $m$ bin as follows: We pick balls
one by one and through them into random bins (independently). Then the number of balls in bins are independent and are distributed according to the Poisson
distribution with parameter $\lambda/m$.
\end{fact}

In the model with outlier edges, $D_2$ is the distribution of the random variable $Z_{P_1} +\hat{\kappa}(Z_{P_1})$,
where $P_1$ is a Poisson random variable with parameter $bM$ (see Section~\ref{sec:sr-adv}). Since $Y_1,Y_2\sim D_2$,
we may assume that $Y_1 = Z_{L'} +\hat{\kappa}(Z_{L'})$ and $Y_2 = Z_{R'} +\hat{\kappa}(Z_{R'})$ for some
Poisson random variables $Z_{L'}$ and $Z_{R'}$  with parameter $bM$. If the test procedure
added $Z_{L'}$ and $Z_{R'}$ edges between $A$ and $R''$ and between $B$ and $L''$, it would get a graph from
the pure Stochastic Block Model with the planted cut $(A\cup L'', B\cup R'')$. But adding extra
$\hat{\kappa}(Z_{L'})$ and $\hat{\kappa}(Z_{L'})$ edges it gets a graph from the SBM with outlier edges.

We showed that if $(A\cup L'', B\cup R'')$ is the planted cut, then $X_1,X_2\sim D_1$ and $Y_1,Y_2\sim D_2$;
if $(B\cup L'', A\cup R'')$ is the planted cut, $X_1,X_2\sim D_2$ and $Y_1,Y_2\sim D_1$. This the restricted classifier outputs
the correct cut with probability $1 - \delta$, this test errs also with probability $\delta$.
\end{proof}

We will need the following simple lemma.

\begin{lemma}\label{lem:coupling-to-bound}
Consider two distributions $D_1$ and $D_2$. Suppose that there exists a joint distribution $D_{12}$ of random variables $X$ and $Y$ such that
$X\sim D_1$ and $Y\sim D_2$, and
$$\Pr(X = Y)\geq \eta.$$
Then, for any test for pairs of distributions $D_1$, $D_2$ (see above) errs with probability at least $\eta^4/2$.
\end{lemma}
\begin{proof}
Consider four independent pairs of random variables $(X_1,Y_1)$, $(X_2,Y_2)$, $(X_3,Y_3)$, and $(X_4,Y_4)$.
Each pair $(X_i,Y_i)$ is distributed according to $D_{12}$. Let $\zeta$ be the error probability of $T$.
Consider two experiments: In the first experiment we apply the test to the pairs $(X_1,Y_2)$ and $(X_3,Y_4)$; in the second, we apply the test
to $(Y_1,X_2)$ and $(Y_3,X_4)$. Observe that the random variables $X_1$, $Y_2$, $X_3$ and $Y_4$ are independent;
and  $X_1\sim D_1$, $Y_2\sim D_2$, $X_3 \sim D_1$, $Y_4\sim D_2$.
The random variables $Y_1$, $X_2$, $Y_3$ and $X_4$ are also independent; but
$Y_1\sim D_2$, $Y_2\sim D_1$, $X_3 \sim D_2$, $Y_4\sim D_1$. So the test should output opposite results in the first and second
experiments. However, with probability at least $\eta^4$, we get $X_1=Y_1$, $X_2=Y_2$, $X_3=Y_3$, $X_4=Y_4$. In this case, the
test returns the incorrect answer either in the first or second experiments.
\end{proof}

We now prove Theorem~\ref{thm:lb}.

\begin{proof}[Proof of Theorem~\ref{thm:lb}]
By Corollary~\ref{lem:pois-couple}, which we prove in the next section, there exists a coupling of two Poisson random variables $P_1$, $P_2$ with
parameters $a$ and $b$, such that
$$\delta \equiv\Pr (P_1 = P_2) \geq C_1 e^{-\frac{C_2(a - b)^2}{a + b}}$$
for some absolute constants $C_1$ and $C_2$. By Lemma~\ref{lem:coupling-to-bound}, the error
probability of any test for $P_1$, $P_2$ is at least $\delta$. Since the number of neighbours
of a fixed vertex $u$ on the same side and on the opposite side are distributed as the Poisson distribution with
parameters $a$ and $b$, by Lemma~\ref{lem:rest-to-test}, we get that
any restricted classifier has error probability at least $\delta$. Finally, by Lemma~\ref{lem:restr-pure},
the expected number of misclassified vertices is at lest
$\delta/2 - O(1/n) = (C_1/2) e^{-\frac{C_2(a - b)^2}{a + b}}-O(1/n)$.
This proves the bound \ref{eq:Bound1}.

In the model with outlier edges, the total number of edges between the set $L'$ and $L''$ has the Poisson distributed
with parameter $(a/n) |L'|\cdot |L\setminus L'| = a \rho (1-\rho)n$.
The total number of edges between $L'$ and $R''$
has the same distribution as $P_1+\hat{\kappa}(P_1)$, where $P_1$ is the Poisson distribution with parameter $b$ (see
Corollary~\ref{cor:kappa}). By Lemma~\ref{lem:coupling-to-bound} and Corollary~\ref{cor:kappa}, the error probability
of any test for these two distributions is at lest 1/2. Hence, by Lemma~\ref{lem:rest-to-test} and Lemma~\ref{lem:restr-sr},
the expected number of misclassified vertices is at least (see Section~\ref{sec:sr-adv})
$$\delta \geq \frac{|L'|}{2n} = \Omega\Big(\frac{\varepsilon (a-b)}{a+b}\Big).$$
This proves the bound \ref{eq:Bound2}.
\end{proof}

\subsection{Poisson Distribution}

\begin{fact}[Median of the Poisson distribution]
For every Poisson random variable $P$ with parameter $\lambda > 0$,
$$\Pr (P \geq \rounddown{\lambda})\geq \frac{1}{2}.$$
\end{fact}

\begin{lemma}\label{lem:Pois-div}
There exists a constant $C > 0$ such that for a Poisson random variable with parameter $\lambda \geq 1$ and every $t\ge 1$, the following
inequality holds:
$$\Pr(P \geq \lambda + t\sqrt{\lambda}) \geq e^{-Ct^2}.$$
\end{lemma}
\begin{proof}
Let $S' = \{k\in \bbZ^+: \rounddown{\lambda} \leq k < \lambda + t\sqrt{\lambda}\}$ and
$S'' = \{k\in \bbZ^+: k \geq \lambda + t\sqrt{\lambda}\}$. The union $S'\cup S''$ is the set of all integers greater than $\rounddown{\lambda}$.
Hence,
$$\Pr (P \in S'\cup S'') = \Pr (P\geq \rounddown{\lambda})\geq 1/2.$$
If $\Pr (P \in S')\leq 1/4$, then $\Pr (P \in S'')\geq 1/4$, and we are done. So we assume that $\Pr (P \in S')\geq 1/4$.
Let $\Delta = \roundup{t\sqrt{\lambda}} + 1$. Notice that $S' + \Delta \equiv \{k+\Delta: k\in S'\}\subset S''$, and, consequently,
$\Pr (P\in S'')\geq \Pr (P \in S' +\Delta)$. We lower bound $\Pr (P \in S' +\Delta)$ using the following lemma.
\begin{lemma}
Let $S$ be a subset of natural numbers. Suppose that all elements in $S$ are upper bounded by $K$. Then,
$$\frac{\Pr(P\in S)}{\Pr(P\in S+1)}\leq 1 + \frac{K-\lambda + 1}{\lambda},$$
where $P$ is a Poisson random variable with parameter $\lambda$.
\end{lemma}
\begin{proof}
Write,
$$\frac{\Pr(P\in S)}{\Pr(P\in S+1)} = \frac{\sum_{k\in S}\Pr(P=k)}{\sum_{k\in S}\Pr(P=k+1)}\leq \max_{k\in S}\frac{\Pr(P=k)}{\Pr(P=k+1)}.$$
For each $k\in S$, we have
$$\frac{\Pr(P = k)}{\Pr (P = k+1)}= \frac{e^{-\lambda}\lambda^k/k!}{e^{-\lambda}\lambda^{k+1}/(k+1)!} = \frac{k+1}{\lambda}\leq \frac{K+1}{\lambda}.$$
Hence,
$$\frac{\Pr(P\in S)}{\Pr(P\in S+1)}\leq 1 + \frac{K-\lambda + 1}{\lambda}$$
\end{proof}

Applying this lemma $\Delta$ times to the set $S$ with $K=\lambda + 2\Delta$, we get
$$\frac{\Pr (P \in S')}{\Pr (P \in S' +\Delta)}\leq \Big(1 + \frac{2\Delta + 1}{\lambda}\Big)^{\Delta}=
\exp\Big(\Delta\ln\Big(1 + \frac{2\Delta + 1}{\lambda}\Big)\Big)\leq \exp \Big(\frac{\Delta (2\Delta + 1)}{\lambda}\Big).$$
Since $\Pr(P \in S')\geq 1/4$ and $\Delta = \roundup{t\sqrt{\lambda}} + 1$, we get for constant $C > 0$,
$$\Pr (P \in S' +\Delta)\geq
\frac{e^{-\frac{2\Delta^2 + \Delta}{\lambda}}}{4}\geq
\frac{e^{-3\Delta^2/\lambda}}{4}\geq e^{-Ct^2}.$$
This finishes the proof.
\end{proof}

\begin{corollary}[Coupling of two Poisson random variables]\label{lem:pois-couple}
There exists positive constants $C_1, C_2>0$ such that for all positive $\lambda_1$ and $\lambda_2$, there exists a joint distribution
of two Poisson random variables $P_1$ and $P_2$ with parameters $\lambda_1$ and $\lambda_2$ such that
$$\Pr (P_1 = P_2) \geq C_1 e^{-\frac{C_2(\lambda_1 - \lambda_2)^2}{\lambda_1 + \lambda_2}}.$$
\end{corollary}
\begin{proof}
Consider the coupling of $P_1$ and $P_2$ that maximizes the probability of the event $\{P_1=P_2\}$. The probability
that $P_1$ and $P_2$ are equal can be expressed in terms of the total variation distance between the distributions of $P_1$ and $P_2$:
$$\Pr(P_1=P_2) = 1 - \|P_1 - P_2\|_{TV} = \sum_{k=0}^{\infty} \min (\Pr(P_1=k),\Pr(P_2 = k)).$$
Assume without loss of generality that $\lambda_1\leq \lambda_2$. We now consider several cases.

\medskip

\noindent I. If $\lambda_1\geq 1$ and $\lambda_2 \leq 2 \lambda_1$, then
$$\Pr(P_1=P_2) \geq \sum_{k> \lambda_2}^{\infty} \min (\Pr(P_1=k),\Pr(P_2 = k))
= \sum_{k> \lambda_2}^{\infty} \Pr(P_1 = k) = \Pr (P_1\ge \lambda_2) .$$
By Lemma~\ref{lem:Pois-div},
$$\Pr(P_1=P_2) \geq  \Pr (P_1\ge \lambda_2) \geq e^{-C\frac{(\lambda_1 - \lambda_2)^2}{\lambda_1}}\geq e^{-9C\frac{(\lambda_1 - \lambda_2)^2}{\lambda_1+\lambda_2}}.$$

\noindent II. If $\lambda_2 \geq 2 \lambda_1$, then
$$\Pr(P_1=P_2) \geq \min (\Pr(P_1=0),\Pr(P_2 = 0)) = \Pr(P_2 = 0) = e^{-\lambda_2} \geq
e^{-\frac{\nicefrac{9}{2}(\lambda_1 - \lambda_2)^2}{\lambda_1+\lambda_2}}.$$
In the last inequality we used that
$$\frac{(\lambda_1 - \lambda_2)^2}{\lambda_1+\lambda_2}\geq \frac{\nicefrac{1}{2} \lambda_2}{\nicefrac{3}{2} \lambda_2} = \frac{2\lambda_2}{9}.$$

 \noindent III. Finally, if $\lambda_1\leq 1$ and $\lambda_2 \leq 2 \lambda_1$, then, as in the previous case,
$$\Pr(P_1=P_2) \geq  \Pr(P_2 = 0) = e^{-\lambda_2} \geq e^{-2}\geq e^{-2}
e^{-\frac{(\lambda_1 - \lambda_2)^2}{\lambda_1+\lambda_2}}.$$
This finishes the proof.
\end{proof}

\begin{lemma}\label{lem:poisson-separation}
For every positive $\lambda_1 \leq \lambda_2$  there exists a joint
distribution of two Poisson random variables $P_1$ and $P_2$ such that
$$\Pr \big(P_2 \geq P_1 \text{ and } P_2 - P_1 \leq 2 (\lambda_2 - \lambda_1)\big)\geq \frac{1}{2}.$$
\end{lemma}
\begin{proof}
Observe that the Poisson distribution with parameter $\lambda_2$ stochastically dominates the Poisson distribution with parameter
$\lambda_1$ (simply because a Poisson random with parameter $\lambda_2$ can be expressed as the sum of two independent Poisson random variables with parameters $\lambda_1$ and $\lambda_2 - \lambda_1$). Thus, there exists a coupling of $P_1$ and $P_2$ such that $P_2\geq P_1$ a.s. We have
$\E[P_2-P_1] = \lambda_2-\lambda_1$, and, by Markov's inequality,
$\Pr ((P_2-P_1)\geq 2(\lambda_1-\lambda_1))\leq 1/2$.
\end{proof}
\begin{corollary}\label{cor:kappa1}
For every positive $\lambda_1 < \lambda_2$, there exists a random function $\kappa: \bbZ^{\geq 0}\to \bbZ^{\geq 0}$
such that $P + \kappa(P)$ has the Poisson distribution with parameter $\lambda_2$ if $P$ has the Poisson distribution
with parameter $\lambda_1$ ($P$ and $\kappa$ are independent).
\end{corollary}
\begin{proof}
Consider Poisson random variables $P_1$ and $P_2$ as in Lemma~\ref{lem:poisson-separation}. Let $\kappa(i) = j$
with probability $\Pr (P_2 = i + j \given P_1 = i)$. Then, clearly, $\kappa (P_1)$ is distributed as $P_2$.
\end{proof}

\begin{corollary}\label{cor:kappa}
For every positive $\lambda_1 < \lambda_2$, there exists a random function $\hat{\kappa}: \bbZ^{\geq 0}\to \bbZ^{\geq 0}$
such that $\hat{\kappa}(P_1)\leq 2(\lambda_2 - \lambda_1)$ a.s. for a Poisson random variable $P_1$ with parameter $\lambda_1$
and there exists a coupled Poisson random variable $P_2$ with parameter $\lambda_2$ such that
$$\Pr (P_2 = P_1+\hat{\kappa}(P_1))\geq 1/2.$$
\end{corollary}
\begin{proof}
We let $\hat{\kappa}(i) = \min (\kappa(i), 2(\lambda_2-\lambda_1))$ and $P_2 = P_1 +\kappa(P_1)$. Then, clearly $\hat{\kappa}(i)\leq 2(\lambda_2-\lambda_1)$ and
$P_1 + \hat{\kappa}(P_1) = P_1 + \kappa(P_1)$ with probability at least $1/2$.
\end{proof}

\end{document}